\documentclass[twoside,11pt]{article}
\usepackage{appendix}
\usepackage[T1]{fontenc}
\usepackage[utf8]{inputenc}
\usepackage{textcomp}

\usepackage{amsthm}
\usepackage{amsmath}
\usepackage{amssymb}
\usepackage{amsfonts}
\usepackage{mathrsfs}
\usepackage{dsfont}

\usepackage[authoryear]{natbib}

\usepackage{enumerate}
\usepackage{mathrsfs}
\usepackage{algorithm}
\usepackage{algorithmicx}
\usepackage{algpseudocode}

\usepackage{array}
\usepackage{caption}
\usepackage{booktabs}
\usepackage{url}

\usepackage{{graphicx}}
\graphicspath{{./fig-JMLR/}}

\usepackage{color}
\usepackage{xcolor}

\DeclareMathOperator{\diag}{diag}

\DeclareMathOperator{\IMSE}{\mathsf{IMSE}}

\DeclareMathOperator{\tr}{trace}

\numberwithin{equation}{section}

\def\ra{\rightarrow}
\def\e1{\mathrm{e}}

\def\0b{\boldsymbol{0}}
\def\1b{\boldsymbol{1}}
\def\TT{^\top}
\def\dd{\mbox{\rm d}}

\def\SB{{\mathscr B}}

\def\SD{{\mathscr D}}

\def\SL{{\mathscr L}}

\def\SP{{\mathscr P}}
\def\cP{{\mathcal{P}}}

\def\SSS{\mathscr S}

\def\SN{{\mathscr N}}

\def\SX{{\mathscr X}}

\def\ma{\alpha}

\def\yb{\boldsymbol{\eta}}
\def\ybb{\overline{\yb}}
\def\mg{\gamma}
\def\mgb{\boldsymbol{\gamma}}
\def\ml{\lambda}
\def\me{\epsilon}

\def\mve{\varepsilon}
\def\mveb{\boldsymbol{\varepsilon}}

\def\ms{\sigma}
\def\mt{\theta}

\def\mtb{\boldsymbol{\theta}}

\def\phib{\boldsymbol{\phi}}
\def\Phib{\boldsymbol{\Phi}}

\def\Lambdab{\boldsymbol{\Lambda}}

\def\Gammab{\boldsymbol{\Gamma}}
\def\mSb{\boldsymbol{\Sigma}}

\def\betab{\boldsymbol{\beta}}

\def\xib{\boldsymbol{\xi}}

\def\Ex{\mathsf{E}}

\def\Ab{\mathbf{A}}
\def\ab{\mathbf{a}}

\def\cb{\mathbf{c}}
\def\cbb{\overline{\cb}}

\def\Cb{\mathbf{C}}
\def\Db{\mathbf{D}}

\def\eb{\mathbf{e}}

\def\Hb{\mathbf{H}}
\def\Ib{\mathbf{I}}

\def\Kb{\mathbf{K}}
\def\Mb{\mathbf{M}}
\def\Ob{\mathbf{O}}

\def\ssb{\mathbf{s}} %renamed sb->ssb (conflict with JMLR style)
\def\Sb{\mathbf{S}}

\def\ub{\mathbf{u}}
\def\Ub{\mathbf{U}}
\def\vb{\mathbf{v}}
\def\Vb{\mathbf{V}}
\def\wb{\mathbf{w}}

\def\xb{\mathbf{x}}
\def\Xb{\mathbf{X}}
\def\yb{\mathbf{y}}
\def\Yb{\mathbf{Y}}
\def\zb{\mathbf{z}}
\def\Zb{\mathbf{Z}}

\theoremstyle{plain}
\newtheorem{theo}{Theorem}[section]
\newtheorem{coro}[theo]{Corollary}
\newtheorem{lemm}[theo]{Lemma}
\newtheorem{remark}{Remark}[section]

\newcommand{\fin} {\mbox{}~\hfill{\lower-0.3ex\hbox{$\triangleleft$}}}

\title{Fast Screening Rules for Optimal Design via\\ Quadratic Lasso Reformulation}

\author{
Guillaume Sagnol\footnote{Institut f\"ur Mathematik, Sekr. MA 5-2, Technische Universit\"at Berlin, Stra{\ss}e des 17.\ Juni 136, 10623 Berlin, Germany; {\tt sagnol@math.tu-berlin.de}} 
\ and \
Luc Pronzato\footnote{Universit\'e C\^ote d'Azur, CNRS, Laboratoire I3S -- 2000 route des lucioles, 06900 Sophia Antipolis, France;
{\tt Luc.Pronzato@cnrs.fr}}
}

\begin{document}

\maketitle

\begin{abstract}%   <- trailing '%' for backward compatibility of .sty file
The problems of Lasso regression and optimal design of experiments
share a critical property: their optimal solutions are typically \emph{sparse}, i.e., only a small fraction of the optimal variables
are non-zero.
Therefore, the identification of the support of an optimal solution reduces the dimensionality of the problem and can yield a substantial simplification of the calculations.
It has recently been shown that linear regression with a \emph{squared} $\ell_1$-norm sparsity-inducing penalty is equivalent to an optimal experimental design problem.
In this work, we use this equivalence to derive safe screening rules that can be used to discard inessential samples. Compared to previously existing rules, the new tests are much faster to compute, especially for problems involving a parameter space of high dimension, and can be used dynamically within any iterative solver, with negligible computational overhead. Moreover,
we show how an existing homotopy algorithm to compute the regularization
path of the lasso method can be reparametrized with respect to the squared $\ell_1$-penalty. This allows the computation of a Bayes $c$-optimal
design in a finite number of steps and can be several orders of magnitude faster than standard first-order algorithms.
The efficiency of the new screening rules and of the homotopy algorithm are demonstrated on different examples based on real data.
\end{abstract}


%\noindent AMS subject classifications:  62K05, 62J07, 90C46

\section{Introduction}\label{S:Intro}

Estimation or prediction of unknown quantities from experimental data are among the most classical problems in statistics and machine learning. \emph{Optimal design of experiments} plays a central role in this process, as it questions \emph{which} data should be collected in order to make estimation/prediction as accurate as possible. In a regression problem, the quality of an experimental design is usually measured through the covariance matrix of the estimator it produces, and a scalar criterion, function of this matrix, is then minimized to provide an $A$-, $c$-, $D$-, $E$- or $L$-optimal design. One can refer, e.g., to~\citet{Fed72,Silvey80,Pukelsheim93} for a thorough exposition of the theory of optimal experiments and a discussion of optimality criteria. Similarly, supervised learning in a classification problem requires labelling samples that are originally unlabelled. This operation may involve taking physical measurements, conducting a poll, running computer simulations or consulting an expert. Here, the aim of optimal design theory is to select \emph{the best possible} subset of the data to be labelled, subject to a budget constraint.
In machine learning, optimal design is also known as \emph{active learning} \citep{cohn1996active}.

In this paper, we focus on optimal design for Bayesian estimation. The design is computed off-line, before the collection of any data. However, we assume that prior information on the quantity to be estimated is available, which means that the approach could also be used in a sequential setting: in that case, the selection of \emph{batches of additional samples} relies on the information gathered through data collection at previous iterations. A central objective of this work is to derive screening rules that can be used to safely discard \emph{useless samples}, and hence speed-up the computation of Bayesian optimal designs, by exploiting the recently discovered connection between $c$-optimal design and a problem similar to Lasso-regression~\citep{SagnolP2019}.

\medskip
For most optimality criteria, the computation of an optimal \emph{exact design}, i.e., a multiset of samples of given cardinality minimizing the optimality criterion, is NP-hard~\citep{welch1982algorithmic,civril2009selecting,vcerny2012two}. To circumvent this issue, standard approaches rely on \emph{approximate design theory}; see, e.g., \citet{Silvey80, Pukelsheim93}. When the set of admissible experimental conditions (the full data set available) is finite, it consists in solving a continuous relaxation of the problem in which the \emph{design} is represented by a vector $\wb$ lying in the probability simplex, such that the weight $w_i$ corresponds to the fraction of experimental resources allocated to the $i$-th sample. Then, various rounding techniques can be used to turn $\wb$ into an exact design satisfying performance guarantees; see, e.g., \citet{pukelsheim1992efficient,singh2020approximation}.

In general, only a few samples contribute to an optimal design $\wb^*$; that is, most of the optimal weights $w_i^*$ equal zero. We say that a sample with $w_i^*=0$ is \emph{inessential}.
Identifying inessential samples is crucial to algorithms that compute optimal designs, since their removal from the set of candidates speeds up subsequent iterations: indeed, the complexity per iteration of design algorithms that operate on a finite set of samples $S$ is typically proportional to the cardinality of $S$.
The idea originated from \citet{Pa03} and \citet{HPa06_SPL} for $D$-optimal design (and for the construction of the minimum-volume ellipsoid containing a set of points). It was further extended to Kiefer's~(1974)\nocite{Kiefer74} $\varphi_p$-criteria \citep{P_SPL-2013}, to $E$-optimal design \citep{HarmanR2019}, and to the elimination of inessential points in the smallest enclosing ball problem \citep{P-OMS2019}. 
A related idea was recently used for the computation of \emph{exact designs},
for criteria used in the field of
active learning~\citep{anstreicher2020efficient,li2022doptimal}.
There, convex relaxations of the problem
with some variables fixed at $0$ or $1$
are solved at internal nodes
of a branch-and-bound tree to obtain certificates that some design points
are inessential, or must be part of the exact optimal design. This permits to accelerate the pruning of the
branch-and-bound tree.

The recent work of~\citet{SagnolP2019} sheds more light on the sparsity of optimal designs, as it shows the equivalence between $c$-optimal design and a variant of the Lasso regression problem in which the sparsity-inducing $\ell_1$-norm penalty is squared.
Analogously, $L$-optimal design is equivalent to group-Lasso regression with a squared $\ell_{1,2}$-norm penalty; see Section~\ref{S:equiv} for more details on the precise meaning of ``equivalent''.
In the field of Lasso-regression, the idea of deriving screening criteria, i.e., inequalities satisfied by points that do not support any optimal solution and can thus be safely eliminated, dates back to \citet{ElGhaouiVR2012}. These screening rules have been refined in subsequent work \citep{XiangR2012, FercoqGS2015, XiangWR2016, NdiayeFGS2017}, with important developments concerning their usage to speed up the solution of a Lasso problem: \citet{BonnefoyERG2015} developed dynamical tests that can be run during the progress of any Lasso algorithm (rather that only at the beginning); \citet{WangWY2015} proposed sequential tests in which a sequence of Lasso problems are solved for decreasing
values of the regularization parameter $\alpha$, the optimal solution
$\xb^*(\alpha_k)$ of the problem for the regularization parameter $\alpha=\alpha_k$
being used to screen-out features in the Lasso problem with regularization parameter $\alpha_{k+1}<\alpha_k$.

\medskip
\emph{Our contribution} in this work is the development of efficient screening rules for $c$-and $L$-optimal designs in the context of Bayesian estimation. In contrast to the elimination rules presented in our previous work~\citep{PS2021-JSPI},
which also work in the absence of prior but
require heavy algebraic computations such as taking matrix square-roots, the present paper leverages the equivalence with a Lasso problem and the geometry of its dual. This yields lightweight screening rules that only require the computation of matrix-vector products and can be used periodically with any algorithm computing an optimal design. While the techniques we use are similar to those used for example in~\citep{FercoqGS2015} for the Lasso problem, the adaptation is not straightforward since the sparsity-inducing penalty is squared, which changes the geometric nature of the dual problem.
We also show that a homotopy algorithm used to compute the regularization
path of a lasso problem~\citep{osborne2000new, efron2004least}
can be adapted to construct a (Bayesian) $c$-optimal design, sometimes much faster than
with state-of-the art algorithms.
Last but not least, the code used in our experiments has been published in the form of a python package (\texttt{qlasso}) and is available at
\url{https://gitlab.com/gsagnol/qlasso}.

The rest of the paper is organized as follows:
Section~\ref{sec:prelim} introduces the $c$- and $L$-optimal
design problems considered in this paper. Section~\ref{sec:Qlasso} presents the quadratic lasso problem and gives precise statements about its equivalence with an optimal design problem. The dual quadratic lasso problem is derived and we show that it can be interpreted as a projection over a polyhedral cone. Our main result is stated in Theorem~\ref{theo:B0} in the form of a general screening rule, which is declined in two corollaries giving inequalities adapted for algorithms iterating either on lasso variables (Corollary~\ref{C:B1}) or on design weights (Corollary~\ref{C:B2}).
These results are adapted to the case of $L$-optimality in Appendix~\ref{S:Lopt}. The adaptation of the homotopy algorithm for the lasso to the case of Bayesian $c$-optimality is derived in Section~\ref{S:homotopy}.
Finally, numerical examples demonstrating the performance of the new screening rules and the homotopy algorithm are presented in Section~\ref{S:examples}.

\section{Background on optimal designs and notation}\label{sec:prelim}

Let $\mathds{H}=\{\Hb_i,\, i=1,\ldots,p\}$ denote a set of $m\times m$ symmetric positive definite matrices, which we shall call elementary information matrices. For any vector of weights $\wb$ in the probability simplex
\begin{align*}
\SP_{p}=\{\wb\in\mathds{R}^{p},\, \wb \geq \0b,\ \sum_{i=1}^p w_i = 1\} \,,
\end{align*}
we denote by $\Mb(\wb)$ the information matrix
\begin{align}\label{M}
\Mb(\wb) = \sum_{i=1}^p w_i\, \Hb_i \,.
\end{align}
For $\cb$ a given vector in $\mathds{R}^m$, a $c$-optimal design is a solution of the optimization problem
\begin{align}\label{phi}
\underset{\wb \in \SP_{p}}{\min}\ \ \phi_c(\wb) = \Phi_c[\Mb(\wb)] = \cb\TT \Mb^{-1}(\wb) \cb \,,
\end{align}
where the decision variable $\wb$ defines a probability measure over the finite space $[p]=\{1,\ldots,p\}$ and is called \emph{design}.
We shall denote $\wb^*$ a $c$-optimal design. Note that nor $\wb^*$ neither $\Mb(\wb^*)$ are necessarily unique, but Theorem~\ref{T:dual-optimal} shows that $\Mb^{-1}(\wb^*)\cb$ is unique.
More generally, we also consider $L$-optimal designs that minimize a linear optimality criterion: given an $m\times m$ positive semidefinite matrix $\Cb=\Kb\Kb\TT\succeq 0$, an $L$-optimal design solves
\begin{align}\label{phiL}
\underset{\wb \in \SP_{p}}{\min}\
\phi_L(\wb)  = \Phi_L[\Mb(\wb)] = \tr[\Cb \Mb^{-1}(\wb)] = \tr[\Kb\TT \Mb^{-1}(\wb)\, \Kb]\,.
\end{align}
Note that $L$-optimality contains $c$-optimality (for $\Kb=\cb$, see~\eqref{phi}), and $A$-optimality (for $\Kb=\Cb=\Ib_m$) as special cases (hence $L$-optimality is also called $A_{\Kb}$-optimality by some authors).

These problems arise in the context of Bayesian estimation for
the linear model
\begin{align}\label{regression1}
Y_i=\ab_i\TT\mtb+\mve_i \,\quad \forall i\in[p],
\end{align}
where $\ab_i\in\mathds{R}^m$ is a sample of $m$ features,
$Y_i$ is the (noisy) observation for the $i$th sample,
and the errors $\mve_i$ are mutually independent and normal $\SN(0,\ms^2)$. We further assume that a prior $\SN(\mtb^0,\mSb)$ on the vector of unknown parameters $\mtb$ is available, with $\mSb$ an $m\times m$-positive definite matrix (which we write $\mSb\succ 0$).

Suppose that $n$ observations are collected according to the design $\wb\in\SP_{p}$, with $\wb$ such such that $n\, w_i\in\mathds{Z}_{\geq 0}$ for all $i$;
that is, for each $i$ we collect $n_i=n\, w_i$ independent observations $Y_{i,j}$, $j\in[n_i]$. Then, the posterior variance of $\cb\TT\mtb$
is
$\frac{\sigma^2}{n}\cdot \cb\TT\left( \sum_{i} w_i \ab_i\ab_i\TT +\frac{\sigma^2}{n}\mSb^{-1}\right)^{-1}\cb$.
Minimizing this posterior variance is equivalent to minimizing ${\cb'}\TT\Mb^{-1}(\wb) \cb'$  where $\cb'=\mSb^{1/2}\cb$,
$\Mb(\wb) = \sum_{i=1}^p w_i\, \Hb_i$ and $\Hb_i= \ab_i' {\ab_i'}\TT + \lambda\, \Ib_m$ for all $i$, with $\ab_i'=\mSb^{1/2}\ab_i$ and $\lambda = \sigma^2/n$. We can thus assume without any loss of generality that $\mSb=\Ib_m$ and that $\Hb_i$ is of the form
\begin{align}\label{Hi}
\Hb_i=\ab_i\ab_i\TT + \ml\, \Ib_m\,, \ \ml>0\,.
\end{align}

Following the same lines as above, we see that minimizing $\phi_L(\wb)$ is equivalent to minimizing the sum of posterior variances of the vector $\Kb\TT \mtb$ in the regression model with Gaussian prior.
One may refer to \citet{Pilz83} for a book-length exposition on optimal design for Bayesian estimation (Bayesian optimal design).
It is noteworthy that
$L$-optimality finds other applications than in design for parameter estimation. For instance, $A$-optimality can be used to construct space-filling designs, by kernel-based \citep{GP-CSDA2016} or geometrical \citep{PZ2019-JSC} approaches. In \citep{P-RESS2019}, $L$-optimal designs are also used to estimate Sobol' indices for sensitivity analysis.
An example of space-filling based on $L$-optimality with $p=2^{13}=8\,192$ and $m=50$ is presented in Section~\ref{S:IMSE}.

In the rest of the paper we adopt the machinery of \emph{approximate design theory} and ignore the integrality requirements $n\, w_i\in\mathds{Z}_{\geq 0}$. We thus obtain the convex optimization problems~\eqref{phi} and~\eqref{phiL} for $c$- and $L$-optimality, respectively. Various algorithms have been proposed to solve these problems, starting from traditional methods
such as 
vertex-direction algorithms~\citep{Wyn70,Fed72},
which are adaptations of the celebrated Frank-Wolfe method,
and multiplicative weight update algorithms~\citep{Fellman74,Yu2010AoS}. Another approach is to reformulate these problems as second-order cone programs, which can be handled by interior-point solvers~\citep{Sagnol2011}.
Recent progress has been obtained through randomization~\citep{harman2020randomized}, or through a reformulation as an unconstrained problem with squared lasso penalty, for which algorithms such as FISTA or block coordinate descent can be used~\citep{SagnolP2019}.
The screening rules presented in this paper can be used to
speed-up any of the aforementioned algorithms. These rules
are in the form of inequalities that are satisfied by any
inessential sample $\ab_i$ such that $w_i^*=0$ for all optimal designs $\wb^*$. When this occurs, we also say that the design point $\ab_i$ or the matrix $\Hb_i$  cannot \emph{support} an optimal design.

We denote respectively by $\|\xb\|=$ $\left(\sum_{i=1}^p x_i^2\right)^{1/2}$, $\|\xb\|_1=\sum_{i=1}^p |x_i|$ and $\|\vb\|_\infty=\max_{i\in[p]}|x_i|$
the $\ell_2$-, $\ell_1$- and $\ell_\infty$-norm of a vector $\xb\in\mathds{R}^p$; $\eb_i$ denotes the $i$-th canonical basis vector of $\mathds{R}^p$. For a $p\times r$ matrix $\Xb$, we denote by $\|\Xb\|_F=[\tr(\Xb\TT\Xb)]^{1/2}$ the Frobenius norm of $\Xb$ and by $\|\Xb\|_{1,2}=\sum_{i=1}^p \|\Xb_{i,\cdot}\|=\sum_{i=1}^p \left(\sum_{j=1}^r \Xb_{i,j}^2\right)^{1/2}$ its $\ell_{1,2}$-norm, with $\Xb_{i,\cdot}$ denoting the $i$th row of $\Xb$; $\SB_m(\xb,\rho)$ denotes the closed Euclidean ball with center $\xb$ and radius $\rho$ in~$\mathds{R}^m$.

\section{Optimal design and (quadratic) Lasso} \label{sec:Qlasso}

\subsection{Equivalence of Optimal Design and Quadratic Lasso}\label{S:equiv}

\citet{SagnolP2019} have shown the equivalence between $c$- (respectively, $L$-) optimal design and a quadratic Lasso (respectively, group-Lasso) problem.

We present below a simplified proof of this result for the case of $c$-optimality, which also leads to a more precise statement. The extension of this result to the case of $L$-optimality is carried out in Appendix~\ref{S:Lopt}.

\begin{theo}\label{theo:eq-copt-qlasso}
Let $\Hb_i$ be given by~\eqref{Hi} and denote by $\Ab$ the $m \times p$ matrix with columns $(\ab_i)_{i\in[p]}$. Then, the $c$-optimal design problem~\eqref{phi} is \emph{equivalent} to the following problem, which we call \emph{quadratic Lasso}:
\begin{align}\label{qlasso}
\min_{\xb \in \mathds{R}^{p}}\quad \SL_\lambda(\xb)= \| \Ab\xb-\cb\|^2 + \lambda\, \|\xb\|_1^2,
\end{align}
in the following sense:
\begin{itemize}
 \item[(i)] the optimal value of~\eqref{qlasso} is equal to $\lambda\, \phi_c(\wb^*)$, where $\wb^*$ is a $c$-optimal design;
 \item[(ii)] If $\xb^*$ solves the quadratic Lasso problem and $\Ab\TT\cb\neq\0b$,  then $\xb^*\neq\0b$ and $\widehat\wb^*=\widehat\wb(\xb^*)$ is $c$-optimal, where
\begin{align}\label{hatw}
\widehat w_i(\xb)=\frac{|x_i|}{\|\xb\|_1}\,, \ \forall i \in[p]\,, \ \xb\neq\0b\,;
\end{align}
 \item[(iii)] If $\wb^*\in\SP_{p}$ is $c$-optimal, then $\widehat\xb^*=\widehat\xb(\wb^*)$ is optimal for~\eqref{qlasso}, where
\begin{align}
\widehat x_i(\wb)= w_i\, \ab_i\TT \Mb^{-1}(\wb)\cb\,, \ \forall i\in[p]\,. \label{widehat-x}
 \end{align}
 \item[(iv)]  In the pathological case $\Ab\TT\cb=\0b$, the unique optimal solution to the quadratic lasso is $\xb^*=\0b$, while every design $\wb\in\SP_p$ is $c$-optimal.
 \end{itemize}
\end{theo}

\begin{proof}
We introduce the function $v: \mathds{R}^p \times \SP_p\to\mathbb{R}\cup\{\infty\}$
by
\begin{align}\label{def_v}
v(\xb,\wb) = \|\Ab\xb-\cb\|^2 + \lambda\, \sum_{i=1}^p \frac{x_i^2}{w_i},
\end{align}
where the function is defined by continuity when $w_i=0$, that is, we assume that $\frac{x_i^2}{0}=0$ if $x_i=0$ and $\frac{x_i^2}{0}=\infty$
otherwise~\footnote{As shown in~\cite{SagnolP2019}, $v(\xb,\wb)$ actually represents the variance of the unbiased linear estimator $\hat{\boldsymbol{\xi}}=\xb^{\raisebox{-0.2mm}{{\tiny $\top$}}} Y + (\cb -\Ab^{\raisebox{-0.2mm}{{\tiny $\top$}}}\xb)^{\raisebox{-0.2mm}{{\tiny $\top$}}} Z$ for $\boldsymbol{\xi}=\cb\TT\boldsymbol{\theta}$,
in the model with observations $Y=\Ab\TT\boldsymbol{\theta}+\boldsymbol{\epsilon}$,
$Z=\boldsymbol{\theta}+\boldsymbol{\nu}$
and errors
$\boldsymbol{\epsilon}\sim\mathcal{N}(\0b,\ml\operatorname{Diag}(\{w_1,\ldots,w_p\})^{-1})$,
$\boldsymbol{\nu}\sim\mathcal{N}(\0b,\Ib_m)$, with $\boldsymbol{\epsilon}$
and $\boldsymbol{\nu}$ mutually independent.}. We note that $v$ is convex, as $\xb\mapsto \|\Ab\xb-\cb\|$ is convex and
$(x_i,w_i)\mapsto \frac{x_i^2}{w_i}=w_i\cdot\big(\frac{x_i}{w_i}\big)^2$ is the \emph{perspective function}
of $x_i\mapsto x_i^2$; see, e.g.,~\cite{BoydV2004}.

We next use the fact that
that for all $\xb\in\mathds{R}^{p}\setminus\{\0b\}$, the function $\wb \mapsto \sum_i x_i^2/w_i$ is minimized over $\SP_{p}$ for $\wb=\widehat \wb(\xb)$ given by \eqref{hatw}. Therefore,
it holds
\begin{equation}
 \min_{\wb\in\SP_p}\ v(\xb,\wb)
 =  v(\xb,\widehat{\wb}(\xb))
 = \|\Ab\xb-\cb\|^2 + \lambda\, \sum_{i=1}^p \frac{x_i^2}{|x_i|} \cdot \|\xb\|_1
 = \SL_\ml(\xb),\label{minv_w}
\end{equation}
and the equation $\min_{\wb\in\SP_p}\ v(\xb,\wb)=\SL_\ml(\xb)$
remains valid for $\xb=\0b$ (with $\SL_\ml(\0b)=\|\cb\|^2$).

On the other hand, we can also minimize $v(\xb,\wb)$ with respect
to $\xb$ for a fixed $\wb\in\SP_p$.
If  $w_i = 0$, then all minimizers
$\widehat{\xb}$ of $v(\xb,\wb)$ must satisfy $x_i=0$, as
otherwise $v(\xb,\wb)=\infty$. We thus restrict to
optimizing the other coordinates of $x_i$,
which corresponds to minimizing a function of the
same form as $\xb\mapsto v(\xb,\wb)$, but for some $\wb\neq\0b$. We thus assume w.l.o.g.\ that $\wb\neq\0b$.
Denote by $\Db(\wb)=\diag\{w_1,\ldots,w_p\}$.
The problem to solve is in fact a least square problem,
as
\[
 v(\xb,\wb) = \|\Ab\xb-\cb\|^2 + \ml \xb\TT \Db^{-1}(\wb) \xb  \,.
\]
The unique minimizer is thus
\begin{align*}
\widehat\xb= \left[\Ab\TT\Ab+\ml\Db^{-1}(\wb)\right]^{-1} \Ab\TT\cb &= \frac{1}{\ml}\, \left[\Ib_m - \Db(\wb)\Ab\TT\Mb^{-1}(\wb)\Ab\right]\Db(\wb)\Ab\TT \cb\\
&= \Db(\wb)\Ab\TT \Mb^{-1}(\wb) \cb \,,
\end{align*}
where we have used the Sherman-Morrison-Woodbury
identity for the second equality
and the definition of $\Mb(\wb)$ which gives $\Ab \Db(\wb) \Ab\TT = \Mb(\wb) -\lambda \Ib$
for the last one.
Therefore, $\widehat x_i=w_i\, \ab_i\TT\Mb^{-1}(\wb)\cb$.
Note that this formula is also valid for the coordinates
$i\in[p]$ with $w_i=0$, as we set
$\widehat x_i=0$ in this case. Thus,
$\widehat\xb=\widehat\xb(\wb)$, as given by~\eqref{widehat-x}.

Substitution in $v(\xb,\wb)$ yields
\begin{align}
 \min_{\xb\in\mathds{R}^p} v(\xb,\wb) = v(\widehat{\xb}(\wb),\wb)
 &=
 \|\Ab \Db(\wb) \Ab\TT \Mb^{-1}(\wb) \cb - \cb\|^2 + \lambda\sum_{i=1}^p w_i (\ab_i\TT \Mb^{-1}(\wb) \cb)^2\nonumber\\
 &=
 \|\ml \Mb^{-1}(\wb)\cb\|^2 + \lambda \sum_{i=1}^p w_i \cdot \cb^T \Mb^{-1}(\wb) \ab_i \ab_i\TT \Mb^{-1}(\wb) \cb\nonumber\\
 &=\lambda \cdot \big[ \cb\TT \Mb^{-1}(\wb)(\lambda \Ib_m+\sum_{i=1}^p w_i \ab_i \ab_i\TT) \Mb^{-1}(\wb) \cb\big]\nonumber\\
 & = \lambda \phi_c(\wb),\label{minv_x}
\end{align}
where we have used $\Ab\Db(\wb)\Ab\TT + \ml \Ib_m =
\sum_{i=1}^p  w_i \ab_i \ab_i\TT + \ml \Ib_m =
\Mb(\wb)$.

\medskip
Now, the theorem easily follows from the above observations.
For $(i)$, we use that
\begin{equation}\label{double_min}
 \min_{\xb\in\mathds{R}^p} \SL_\ml(\xb)
 =\min_{\xb\in\mathds{R}^p,\wb\in\SP_p} v(\xb,\wb) =
 \min_{\wb\in \SP_p} \ml \phi_c(\wb),
\end{equation}
where the first equality results from~\eqref{minv_w}
and the second one from~\eqref{minv_x}.
For $(ii)$, we first prove that $\xb^*\neq\0b$ holds
for every optimal solution to the quadratic lasso
whenever $\Ab\TT\cb\neq\0b$.
Let $i$ be any index such that $\ab_i\TT\cb\neq 0$.
Then, we define $\xb = \frac{\ab_i\TT\cb}{\|\ab_i\|^2+\lambda} \eb_i$. It holds
\[
 \SL_\lambda(\xb)=
 \left\|\frac{\ab_i\TT\cb}{\|\ab_i\|^2+\lambda} \ab_i-\cb\right\|^2 + \lambda \frac{(\ab_i\TT\cb)^2}{(\|\ab_i\|^2+\lambda)^2}
=\|\cb\|^2 - \frac{(\ab_i\TT\cb)^2}{\|\ab_i\|^2+\lambda}
< \|\cb\|^2 = \SL_{\ml}(\0b),
\]
which shows that $\0b$ cannot be optimal for the quadratic lasso.
Now, let $\xb^*$ be an optimal solution to the quadratic lasso. We have, for any $\wb\in\SP_p$,
\begin{eqnarray*}
&& \ml \phi_c(\wb) \overset{\raisebox{0.4em}{\scriptsize\eqref{minv_x}}}{=} \min_{\xb\in\mathds{R}^p} v(\xb,\wb) \geq \min_{\xb\in\mathds{R}^p}\min_{\wb\in\SP_p} v(\xb,\wb)
\overset{\raisebox{0.4em}{\scriptsize\eqref{minv_w}}}{=} \min_{\xb\in\mathds{R}^p} \SL_\ml(\xb) = \SL_\ml(\xb^*)
\\
%\min_{\xb\in\mathds{R}^p} \SL_\ml(\xb)  = \SL_\ml(\xb^*)
&& \hspace{3cm} \overset{\raisebox{0.4em}{\scriptsize\eqref{minv_w}}}{=} v(\xb^*,\widehat{\wb}(\xb^*))
 \geq \min_{\xb'\in\mathds{R}^p} v(\xb',\widehat{\wb}(\xb^*))
 \overset{\raisebox{0.4em}{\scriptsize\eqref{minv_x}}}{=}
 \ml \phi_c(\widehat{\wb}(\xb^*))\,,
\end{eqnarray*}
which shows that $\widehat{\wb}^*=\widehat{\wb}(\xb^*)$
is $c$-optimal. Similarly, for $(iii)$,
let $\wb^*$ be a $c$-optimal design. Then, for any $\xb\in\mathds{R}^p$,
\begin{eqnarray*}
&& \SL_\ml(\xb) \overset{\raisebox{0.4em}{\scriptsize\eqref{minv_w}}}{=} \min_{\wb\in\SP_p} v(\xb,\wb) \geq \min_{\wb\in\SP_p} \min_{\xb\in\mathds{R}^p} v(\xb,\wb)
\overset{\raisebox{0.4em}{\scriptsize\eqref{minv_x}}}{=} \min_{\wb\in\SP_p} \ml\phi_c(\wb) = \ml\phi_c(\wb^*)
\\
% \min_{\wb\in\SP_p} \ml\phi_c(\wb) = \ml\phi_c(\wb^*)
&& \hspace{3cm} \overset{\raisebox{0.4em}{\scriptsize\eqref{minv_x}}}{=} v(\widehat{\xb}(\wb^*),\wb^*)
 \geq \min_{\wb'\in\SP_p} v(\widehat{\xb}(\wb^*),\wb')
 \overset{\raisebox{0.4em}{\scriptsize\eqref{minv_w}}}{=}
 \SL_\ml(\widehat{\xb}(\wb^*)),
\end{eqnarray*}
and $\widehat{\xb}^*=\widehat{\xb}(\wb^*)$ is thus an optimal
solution of the quadratic lasso problem.

It only remains to handle the pathological case
$(iv)$, that is, $\Ab\TT\cb=\0b$. For any $\xb\neq\0b$, we have
$\SL_\ml(\xb)=\|\Ab\xb-\cb\|^2+\ml\|\xb\|_1^2
=\|\cb\|^2 + \|\Ab\xb\|^2 +\ml\|\xb\|_1^2 >
\|\cb\|^2=\SL_{\ml}(\0b)$,
so the unique optimal solution to~\eqref{qlasso} is $\xb^*=\0b$. On the other hand, by~\eqref{double_min}, $c$-optimal designs
are minimizers of $\wb\mapsto v(\xb^*,\wb)
=v(\0b,\wb)$, which is the constant function
$\wb\mapsto\|\cb\|^2$.
\end{proof}

For all $\xb\in\mathds{R}^p$, observe that
\begin{align}\label{L>phi}
\SL_\ml(\xb)
= v(\xb,\widehat\wb(\xb))
\geq \min_{\xb'\in\mathds{R}^{p}} v(\xb',\widehat\wb(\xb)) = \ml\, \phi_c(\widehat\wb(\xb)) \,,
\end{align}
with equality when $\xb$ is such that $\widehat\wb(\xb)$ corresponds to $c$-optimal weights $\wb^*$, and therefore when $x_i=\widehat{x}_i(\wb^*)=w_i^*\, \ab_i\TT\Mb^{-1}(\wb^*)\cb$, $\forall i\in [p]$. In other words, both minimization
problems $\min_\xb \SL_\ml(\xb)$ and $\min_\xb \lambda \phi_c(\widehat{\wb}(\xb))$ share the same optimal value $\lambda \phi_c(\wb^*)$, but the former objective function dominates the latter.

\begin{remark}\label{R:multiplicative}
By alternating minimization of~\eqref{def_v} with respect to $\wb$ and $\xb$, with $\wb^k \ra \xb^{k+1}=\widehat\xb(\wb^k) \ra \wb^{k+1}=\widehat\wb(\xb^{k+1}) \ra \cdots$, where $\widehat\wb(\xb)$ and $\widehat\xb(\wb)$ are respectively defined by \eqref{hatw} and \eqref{widehat-x}, we obtain that $\wb^{k+1} = \widehat\wb(\widehat\xb(\wb^k))$ is given by
\begin{align*}
w_i^{k+1}= \frac{w_i^k \, \left|\ab_i\TT\Mb^{-1}(\wb^k)\cb\right|}{\sum_{j=1}^{p} w_j^k\, \left|\ab_j\TT\Mb^{-1}(\wb^k)\cb\right|} \,.
\end{align*}
Since we only consider weights that sum to one, we may rewrite $\phi_c(\wb)$ as $\phi_c(\wb)=\phi'_c(\wb)=\cb\TT\left[\left(\sum_{i=1}^p w_i\,\ab_i\ab_i\TT\right)+\ml \Ib_m\right]^{-1}\cb$. Denote by $\nabla\phi'_c(\wb)$ the gradient of $\phi'_c(\cdot)$ at $\wb$. Since its $i$-th component $\{\nabla \phi'_c(\wb)\}_i$ is equal to $-\cb\TT\Mb^{-1}(\wb)\ab_i\ab_i\TT\Mb^{-1}(\wb)\cb=-(\ab_i\TT\Mb^{-1}(\wb)\cb)^2$, the alternate minimization algorithm above corresponds to
\begin{align}\label{multiplicative}
w_i^{k+1}= \frac{w_i^k \, \left|\{\nabla \phi'_c(\wb^k)\}_i\right|^{1/2}}{\sum_{j=1}^{p} w_j^k \, \left|\{\nabla \phi'_c(\wb^k)\}_j\right|^{1/2}} \,,
\end{align}
which coincides with a variant of the multiplicative weight update algorithm of \citet{Fellman74} for the minimization of $\phi'_c(\wb)$; see also \citet{Yu2010AoS}.
\fin
\end{remark}

%------------------------------------
\subsection{The dual quadratic Lasso}\label{S:dual-lasso}

\medskip
{\sloppy
The primal corresponds to \eqref{qlasso}; that is, $\min_{\xb\in\mathds{R}^{p}} \SL_\ml(\xb)$. Introducing an auxiliary variable \mbox{$\zb=\Ab\xb-\cb$}, we can write this problem as a saddle point problem
\begin{align}\label{Lagrangian}
\min_{\xb\in\mathds{R}^{p},\, \zb\in\mathds{R}^m}\, \max_{\yb\in\mathds{R}^m} \ml\,\|\xb\|_1^2+ \|\zb\|^2 +2\yb\TT(\zb-\Ab\xb+\cb) \,,
\end{align}
and therefore the Lagrangian dual problem reads
\begin{align}\label{primal1}
\max_{\yb\in\mathds{R}^m} \left[ 2\yb\TT\cb + \min_{\zb\in\mathds{R}^m}\, \left(\|\zb\|^2+2\yb\TT\zb\right) +  \min_{\xb\in\mathds{R}^{p}}\,  \left( \ml\,\|\xb\|_1^2-2\yb\TT \Ab\xb \right) \right]\,.
\end{align}
}

Now, recall that the convex conjugate of a function
$f:\mathds{R}^p\to\mathds{R}$ is
$f^*: \yb\, \mapsto \sup_{\xb}\quad \xb\TT \yb - f(\xb)$,
so the above problem can be rewritten as
\begin{align*} %\label{dual_conjugates}
\max_{\yb\in\mathds{R}^m} \left[ 2\, \yb\TT\cb
- 2 g_2^*(-\yb)
- 2 \lambda\, g_1^*\Big(\frac{\Ab\TT\yb}{\lambda}\Big)
\right],
\end{align*}
where $g_1(\xb)= \|\xb\|_1^2/2$ and $g_2(\xb)= \|\xb\|^2/2$. The dual problem then takes a simple form, noting that for any norm $\|\cdot\|$, the convex conjugate of $g: \xb\mapsto \|\xb\|^2/2$ is $g^*: \yb\mapsto \|\yb\|_*^2/2$, where $\|\cdot\|_*$
is the dual norm of $\|\cdot \|$; see~\citet[Chapter~3]{BoydV2004}. Hence the expression to be maximized in \eqref{primal1} is
$2\, \yb\TT\cb - \|\yb\|^2 - \|A\TT \yb\|_\infty^2/\ml$.
%where $\|\zb\|_\infty=\max_i |z_i|$ denotes the $\ell^\infty$ norm of $\zb$.
To summarize, since $2\yb\TT\cb-\|\yb\|^2=\|\cb\|^2-\|\yb-\cb\|^2$, the dual (quadratic) Lasso is
\begin{align}\label{dual}
\max_{\yb\in\mathds{R}^m}\ \SD_\ml(\yb) = \|\cb\|^2-\|\yb-\cb\|^2 - \frac{\|\Ab\TT\yb\|_\infty^2}{\ml}\,.
\end{align}
Moreover, there is no duality gap, since
both the primal and dual problems are unconstrained, and
Slater's constraint qualification trivially holds (see, e.g., \citet{BoydV2004}):
\begin{align}
\min_{\xb\in\mathds{R}^{p}} \SL_\ml(\xb) = \max_{\yb\in\mathds{R}^m} \SD_\ml(\yb)\,.
\end{align}

%------------------------------------
\subsection{Dual-optimal solution}

The following theorem gives a necessary and sufficient condition for a pair $(\xb^*,\yb^*)$ to be primal-dual optimal for the quadratic Lasso.

\begin{theo}\label{T:dual-optimal}
A pair $(\xb^*,\yb^*)\in \mathds{R}^p \times \mathds{R}^m$ is
primal-dual optimal, i.e., $\SD_\ml(\yb^*)=\SL_\ml(\xb^*)$, if and only if
\begin{align}\label{KKT-system}
\yb^*=\cb-\Ab\xb^*, \
\|\Ab\TT \yb^*\|_\infty = \lambda \|\xb\|_1,
\text{ and } \forall i\in[p],
\left\{
\begin{array}{ll}
 |\ab_i\TT\yb^*| \leq \lambda \|\xb^*\|_1 & \text{ if } x_i^*=0;\\
 \ab_i\TT\yb^* =\lambda\, \operatorname{sign}(x_i^*) \|\xb^*\|_1
 & \text{ otherwise}.
\end{array}
\right.
\end{align}
In particular,
\begin{align}\label{likeET}
x_i^*\, (\|\Ab\TT \yb^*\|_\infty - |\ab_i\TT \yb^*|) = 0 \mbox{ for all } i \in [p] \,.
\end{align}
Moreover, the dual optimal point $\yb^*$ is unique and satisfies $\yb^* = \lambda \Mb_*^{-1} \cb$, where $\Mb_*= \Mb(\widehat{\wb}(\xb^*))$, with $\widehat{\wb}(\xb)$ given by \eqref{hatw}.
\end{theo}

\begin{proof}
(\textit{i}) \emph{Uniqueness of the dual optimal solution}. The dual problem $\max_{\yb\in\mathds{R}^m} \SD_\ml(\yb)$ is equivalent to
\begin{align}\label{conic_form}
\min_{\yb\in\mathds{R}^m,\, \|\Ab\TT\yb\|_\infty \leq \sqrt{\ml}u} \|\yb-\cb\|^2 + u^2
=\min_{\ybb\in\cP_\ml(\Ab)}  \|\ybb-\cbb\|^2 \,,
\end{align}
where $\ybb=(\yb\TT,\ u)\TT$ and $\cbb=(\cb\TT,\ 0)\TT$ belong to $\mathds{R}^{m+1}$ and $\cP_\ml(\Ab)$ is the polyhedral cone
\begin{align}\label{Plambda}
\cP_\ml(\Ab) = \{\ybb=(\yb\TT,\ u)\TT\in\mathds{R}^{m+1}: |\ab_i\TT\yb| \leq \sqrt{\ml}\,u \mbox{ for all } i=1,\ldots,m\} \,,
\end{align}
showing that the optimal solution for the dual corresponds to the unique orthogonal projection of $\cbb$ onto $\cP_\ml(\Ab)$; see Figure~\ref{F:dual} for an illustration. %\todo{include a picture}

\noindent (\textit{ii}) \emph{Optimality conditions}.
The optimization problem is convex and Slater's condition holds, so
the Karush-Kuhn-Tucker (KKT) conditions are necessary and sufficient
to characterize a pair $(\xb^*,\yb^*)$ of primal-dual optimal solutions.
The primal problem consists of minimizing $2(\lambda g_1(\xb) + g_2(\zb))$
under the constraint $\zb=\Ab\xb-\cb$, so the KKT system reduces to primal
feasibility, i.e., $\zb^*=\Ab \xb^*-\cb$, and stationarity of the Lagrangian~\eqref{Lagrangian}, i.e.,
$$
\0b \in \partial \left. \left( \xb,\zb \mapsto 2(\lambda g_1(\xb) + g_2(\zb) +  \yb\TT (\zb-\Ab\xb + \cb))\right)\right|_{\xb=\xb^*, \zb=\zb^*}.
$$
Differentiating with respect to $\zb$ gives
$\nabla_{\zb}(\frac{1}{2}\|\zb\|^2+\yb^{*\top}\zb)_{\mid\zb=\zb^*}=\zb^* + \yb^*=\0b$, which already shows that $\yb^*=-\zb^*=\cb-\Ab\xb^*$.

As the Lagrangian~\eqref{Lagrangian} is not differentiable with respect to $\xb$, the stationarity condition with respect to $\xb$ is slightly more complicated. We must solve
$$\0b \in \partial(\xb\mapsto \lambda g_1(\xb) - \yb^*{}\TT \Ab \xb)_{|\xb=\xb^*}\,,$$
which is equivalent to $\Ab\TT\yb^*/\lambda \in \partial g_1(\xb^*)$.
Since $\partial g_1(\xb) = \{ \|\xb\|_1\cdot \ub: \ub \in \partial\|\xb\|_1\}$ (see, e.g., \citet{SagnolP2019}) and since $\ub\in\partial\|\xb\|_1$ if and only if $u_i\in[-1,1]$ whenever $x_i = 0$ and $u_i=\operatorname{sign}(x_i)$ otherwise, we get the following equivalent condition:
\begin{align*}
&\exists \ub\in\partial\|\xb^*\|_1:\quad \Ab\TT\yb^* = \lambda \|\xb^*\|_1 \, \ub\\
\iff\ & \forall i\in[p],\quad \big(x_i^*=0\ \wedge\ |\ab_i\TT\yb^*| \leq \lambda \|\xb^*\|_1\big) \quad\text{or}\quad \big(x_i^*\neq 0\ \wedge\ \ab_i\TT\yb^* =\lambda\, \operatorname{sign}(x_i^*)\ \|\xb^*\|_1 \big).
\end{align*}

\noindent(\textit{iii}) \emph{Alternative expression of $\yb^*$}.
We prove that $\Mb_*(\cb-A\xb^*)=\lambda \cb$, which
implies $\cb-A\xb^*=\lambda \Mb_*^{-1} \cb$, as desired.
Let $\Db_*=\Db(\widehat{\wb}(\xb^*))=(1/\|\xb^*\|_1)\diag\{x_1^*,\ldots,x_p^*\}$. We have
\begin{align*}
 \Mb_*(\cb-\Ab\xb^*) - \lambda \cb = (\Mb_*-\lambda \Ib)(\cb-\Ab\xb^*) - \lambda \Ab\xb^* &= \Ab \Db_* \Ab\TT \yb^* - \lambda \Ab \xb^*\\
 &= \Ab [\Db_* \Ab\TT \yb^* - \lambda \xb^* ].
\end{align*}
Finally, the optimality conditions imply
$|x_i^*|\, \ab_i\TT \yb^*=\lambda x_i \|\xb\|_1^*$, $\forall i\in[p]$, hence
$\Db_* \Ab\TT \yb^* - \lambda \xb^*=\0b$.
\end{proof}

\begin{figure}[t]
\begin{center}
\includegraphics[width=.49\linewidth]{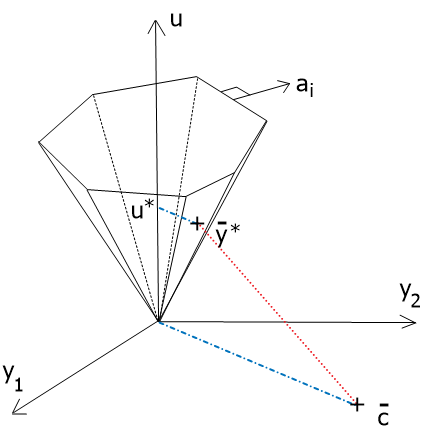}
\end{center}
\caption{\small Uniqueness of the dual optimal solution $\ybb^*$.}\label{F:dual}
\end{figure}

We also recall a well known analogous optimality condition
for the problem of $c$-optimal design.
Based on a stationarity condition,
the so-called \emph{Equivalence Theorem} from optimal design theory
states that $\wb^*\in\SP_p$ is $c$-optimal if and only if the directional derivative of $\phi_c(\cdot)$ in the direction of any
vertex of $\SP_p$ is non-negative; see, e.g., \citet{Silvey80, Pukelsheim93}.

\begin{theo}\label{theo:ET}
Suppose that the $\Hb_i$'s in $\mathds{H}$ satisfy \eqref{Hi}. Then the vector of weights $\wb^*$ is $c$-optimal if and only if
$\cb\TT\Mb^{-1}(\wb^*)(\ab_i\ab_i\TT+\ml\Ib_m)\Mb^{-1}(\wb^*)\cb \leq \phi_c(\wb^*)$ for all $i\in[p]$, with equality for all $i$ such that $w_i^*>0$.
\end{theo}

Note that the theorem implies that $w_i^*(|\ab_i\TT\Mb^{-1}(\wb^*)\cb|-\| \Ab\TT\Mb^{-1}(\wb^*)\cb\|_\infty)=0$ for all $i\in[p]$ when $\wb^*$ is $c$-optimal.
In particular,
$|\ab_i\TT\Mb^{-1}(\wb^*)\cb|$ does not depend
on $i\in[p]$ as soon as $w_i^*>0$, and
the elements of $\widehat{\xb}(\wb^*)$
must be of the form $\widehat{x}_i(\wb^*)=\pm w_i^* \|\Ab\TT\Mb^{-1}(\wb^*)\cb\|_{\infty}$.
Hence,
$|\widehat{x}_i(\wb^*)|$ is proportional to $w_i^*$ and we have
$\widehat{\wb}(\widehat{\xb}(\wb^*))=\wb^*$,
i.e., $\wb^*$ is a fixed point of the mapping $\wb\mapsto\widehat{\wb}(\widehat{\xb}(\wb))$.

\medskip
We next prove an important property about $\widehat \xb(\wb)$.
For $\wb\in\SP_p$ and $\xb\in\mathds{R}^{p}$, we define
\begin{equation}\label{defy1y2}
\yb_1(\xb)=\cb-\Ab\xb\quad  \text{ and }\quad  \yb_2(\wb)=\lambda \Mb^{-1}(\wb)\cb.
\end{equation}
Note that Theorem~\ref{T:dual-optimal} implies
that $\yb_1(\xb^*) = \yb_2(\widehat{\wb}(\xb^*))$ for
any optimal solution $\xb^*$ of the quadratic lasso problem.
The next lemma states that the equation obtained by inverting the role of $\wb$ and $\xb$ holds \emph{for any design} $\wb\in\SP_p$.

\begin{lemm}\label{L:y1=y2}
Let $\wb\in\SP_{p}$. Then,
\begin{align}\label{y1=y2}
\yb_1(\widehat \xb(\wb))=\yb_2(\wb)\,,
\end{align}
where $\widehat \xb(\wb)$ is defined by~\eqref{widehat-x},
and $\yb_1, \yb_2$ are defined by~\eqref{defy1y2}.
\end{lemm}

\begin{proof}
Denote $\Db=\Db(\wb)=\diag\{w_1,\ldots,w_{p}\}$. By definition, $\Mb=\Mb(\wb)=\Ab \Db\Ab\TT+\lambda\Ib_m$ and
\eqref{widehat-x} gives $\widehat \xb(\wb)=\Db\Ab\TT\Mb^{-1}\cb$. It follows that
\begin{align*}
\yb_1(\widehat \xb(\wb)) = \cb-\Ab \widehat \xb(\wb)=\cb-\Ab\Db\Ab\TT\Mb^{-1}\cb=\left[ \Mb-\Ab\Db\Ab\TT\right]\Mb^{-1}\cb=\ml\,\Mb^{-1}\cb=\yb_2(\wb) \,.
\end{align*}\vspace{-2em}
\end{proof}

\medskip
We conclude this section by making explicit
the connection between the quadratic and standard lasso problems. This connection will be exploited in Section~\ref{S:homotopy} to derive a new algorithm
for the computation of $c$-optimal designs.
The (standard) lasso problem reads
\begin{equation}\label{Palpha}
\min_{\xb\in\mathds{R}^p}\quad  \frac12 \|\Ab\xb-\cb\|^2 + \alpha  \|\xb\|_1,
\end{equation}
that is,~\eqref{Palpha} is the Lagrange relaxation of
the minimization of $\frac12 \|\Ab\xb-\cb\|^2$ under the constraint $\|\xb\|_1<t$, or equivalently $\|\xb\|_1^2<t^2$, for some $t>0$. The formulation \eqref{qlasso} with a squared penalty $\|\xb\|_1^2$ thus corresponds to the Lagrangian relaxation of the same problem, with the precise relationship between $t^2$ and $\ml$ depending on $\cb$ and $\Ab$.
(Note that there is a factor $\frac12$ in front of $\|\Ab\xb-\cb\|^2$ in the standard lasso problem~\eqref{Palpha}, but not in~\eqref{qlasso}: this factor is commonly introduced in the standard lasso to have full symmetry with the dual problem; we have refrained from introducing such a factor for the quadratic lasso problem~\eqref{qlasso}, as this would make the formulas of the next section more complex.)

\begin{theo}\label{T:lasso-qlasso}
Let $\alpha < \|\Ab\TT\cb\|_\infty$
and $\xb^*$ be an optimal solution
to the the standard lasso problem~\eqref{Palpha}.
Then, $\xb^*\neq \0b^*$ and $\xb^*$ is an optimal solution to the quadratic lasso
\eqref{qlasso} with $\lambda=\lambda(\alpha)=\alpha/\|\xb^*\|_1$.
Conversely, for all $\lambda>0$, if $\Ab\TT\cb\neq\0b$ and $\xb^*$ is an optimal solution to the quadratic lasso~\eqref{qlasso}, then $\xb^*$ is
also an optimal solution to the standard
lasso~\eqref{Palpha} with $\alpha=\alpha(\ml)=\ml\|\xb^*\|_1$.
\end{theo}

\begin{proof}
 This property is a simple consequence of the KKT conditions of the standard lasso problem~\eqref{Palpha}, which can be found,
 e.g., in~\cite{XiangWR2016}, and read as follows: $\xb^*$ solves~\eqref{Palpha} if and only if there exists a dual vector $\boldsymbol{\theta}^*$ such that
 \begin{equation}\label{KKTalpha}
  \cb=\Ab\xb^* + \alpha \boldsymbol{\theta}^*
  \quad \text{and}\quad
  \forall i\in[p],\ \ab_i\TT \boldsymbol{\theta}^* \in
  \left\{
  \begin{array}{ll}
  \{\operatorname{sign}(x_i^*)\} & \text{ if }x_i^*\neq 0;\\
  {[-1,1]} & \text{ if }x_i^*= 0.
  \end{array}
  \right.
 \end{equation}
First note that $\xb^*=\0b$ can only satisfy \eqref{KKTalpha}
if $\boldsymbol{\theta}^*=\cb/\alpha$, which
implies $\|\Ab\TT\cb\|_\infty=\alpha\|\Ab\TT\boldsymbol{\theta}^*\|_\infty\leq\alpha$.

Now, let $\alpha < \|\Ab\TT\cb\|_\infty$ and $\xb^*\neq \0b$ be an optimal solution to the standard
lasso problem~\eqref{Palpha}, and denote by $\boldsymbol{\theta}^*$ a corresponding optimal dual vector.
Define $\ml = \frac{\alpha}{\|\xb^*\|_1}$, and $\yb^*=\alpha \boldsymbol{\theta}^*=\ml \|\xb^*\|_1 \boldsymbol{\theta}^*$.
Clearly, \eqref{KKTalpha} and $\xb^*\neq\0b$ imply that $(\xb^*,\yb^*)$ fulfills
the KKT system~\eqref{KKT-system}, hence
$\xb^*$ is an optimal solution to the quadratic lasso~\eqref{qlasso}.
Conversely, if
$\Ab\TT\cb\neq\0b$ and
$\xb^*$ solves the quadratic lasso,
we know from Theorem~\ref{theo:eq-copt-qlasso}
that $\xb^*\neq\0b$, so $\alpha=\alpha(\lambda)=\lambda \|\xb^*\|_1> 0$.
Further, if $(\xb^*,\yb^*)$ solves
the KKT equations~\eqref{KKT-system}, then
$(\xb^*,\boldsymbol{\theta}^*)$ solves~\eqref{KKTalpha}
with $\boldsymbol{\theta}^*=\yb^*/\alpha$.
Thus, $\xb^*$ solves the standard lasso with regularizer $\alpha$.
\end{proof}

%------------------------------------
%\subsection{Screening rules for inessential \texorpdfstring{$\ab_i$}{ai}'s}\label{S:rules-r=1}
\subsection{Screening rules for inessential $\ab_i$'s}\label{S:rules-r=1}

The objective is to derive an inequality that must be satisfied by any $\ab_i$ that may support an optimal design.
To this end, we first use information about the current iterate $\xb$ for Problem~\eqref{qlasso} to construct a dual solution $\yb(\xb)$. Then, we derive a bound on $\|\yb(\xb)-\yb^*\|$ and exploit the optimality conditions of the dual problem~\eqref{conic_form}.

\begin{theo}\label{theo:B0}
Let $\yb\in\mathds{R}^m$ be an $\me$-suboptimal dual solution, i.e., $\SD_\lambda(\yb)\geq\SD_\lambda(\yb^*)-\me$.
Then, any $\ab_i$ satisfying
\begin{align*}
D_0(\ab_i; \yb, \me)= \|\Ab\TT\yb\|_\infty- |\ab_i\TT\yb|  - \sqrt{\me(\|\ab_i\|^2+\ml)} >0
\end{align*}
cannot support an optimal design.
\end{theo}

\begin{proof}
We use the same construction as in part~($i$) of the proof of Theorem~\ref{T:dual-optimal}, and  denote $\ybb=(\yb\TT,\ u)\TT$, $\ybb^*=({\yb^*}\TT,\ u^*)\TT$ and $\cbb=(\cb\TT,0)\TT$, with $\yb^*$ the optimal solution of the dual problem~\eqref{dual}, $u=\|\Ab\TT\yb\|_\infty/\sqrt{\ml}$ and $u^*=\|\Ab\TT\yb^*\|_\infty/\sqrt{\ml}$. Then, $\SD_\ml(\yb)=\|\cb\|^2-\|\ybb-\cbb\|^2$ and $\SD_\ml(\yb^*)=\|\cb\|^2-\|\ybb^*-\cbb\|^2$.
The quantity $(\ybb-\ybb^*)\TT(\ybb^*-\cbb)$ is nonnegative,
as $\ybb\in\mathcal{P}_\lambda(A)$ is feasible for the conic formulation of the dual problem~\eqref{conic_form}, and by optimality of $\ybb^*$ the gradient $\nabla\|\ybb-\cbb\|^2_{\mid_{\ybb=\ybb^*}}=2(\ybb^*-\cbb)$ defines a supporting hyperplane to $\mathcal{P}_\lambda(A)$.

If we denote $r=\|\ybb-\cbb\|$ and $r_*=\|\ybb^*-\cbb\|$, then
we can write
$r^2=\|\ybb-\ybb^* + \ybb^*-\cbb\|^2=\|\ybb-\ybb^*\|^2+r_*^2+2(\ybb-\ybb^*)\TT(\ybb^*-\cbb)$, hence $r^2- r_*^2\geq \|\ybb-\ybb^*\|^2$. In addition,
\begin{align*}
r^2-r_*^2=\SD_\ml(\yb^*)-\SD_\ml(\yb)\leq\me\,.
\end{align*}
We thus obtain $\|\ybb-\ybb^*\|^2 \leq r^2-r_*^2 \leq \me$; i.e. $\ybb^* \in \SB_{m+1}(\ybb,\sqrt{\me})$.

Due to the KKT-condition~\eqref{KKT-system} for Problem~\eqref{conic_form}, it holds $x_i^* (|\ab_i\TT \yb^*|-\sqrt{\lambda} u^*)=0$ for all $i\in[p]$. Therefore, any $\ab_i$ that supports an optimal design corresponds to an active constraint at the optimum, and thus satisfies $|\ab_i\TT\yb^*| = \sqrt{\ml}\,u^*$. But $\ybb^*\in\SB_{m+1}(\ybb,\sqrt{\me})$ implies
\begin{align*}
|\ab_i\TT\yb^*| - \sqrt{\ml}\,u^* \leq |\ab_i\TT\yb| - \sqrt{\ml}\,u + \sqrt{\me(\|\ab_i\|^2+\ml)}  \,,
\end{align*}
showing that any $\ab_i$ such that $|\ab_i\TT\yb| < \sqrt{\ml}\,u - \sqrt{\me(\|\ab_i\|^2+\ml)}$ cannot support an optimal design.
\end{proof}

Specializing the above result for the dual solutions $\yb_1(\xb)=\cb-\Ab\xb$ and $\yb_2(\widehat\wb(\xb))=\lambda \Mb^{-1}(\widehat\wb(\xb)) \cb$ (which are reasonable choices by Theorem~\ref{T:dual-optimal}), we obtain the screening criteria given in Corollaries~\ref{C:B1} and \ref{C:B2}.

\begin{coro}\label{C:B1}
Let $\xb\in\mathds{R}^{p}$, and define
\begin{align}
D_1(\ab_i; \xb) &=\ D_0\Big(\ab_i;\ \yb_1,\ \SL_\lambda(\xb)-\SD_\lambda(\yb_1)\Big) \nonumber \\
&=\ \|\vb\|_\infty- |v_i| -  \left( \frac{\|\vb\|_\infty^2}{\lambda} + \lambda\|\xb\|_1^2 - 2\,\xb\TT\vb \right)^{1/2}\, \sqrt{\|\ab_i\|^2+\ml}\,, \label{B1}
\end{align}
where $\yb_1=\yb_1(\xb)=\cb-\Ab\xb$ and $\vb=\vb(\xb)=\Ab\TT\yb_1(\xb)$. Then, any $\ab_i$ satisfying $D_1\big(\ab_i;\, \xb \big)>0$ cannot support an optimal design.
\end{coro}
\begin{proof}
The result directly follows from Theorem~\ref{theo:B0}. The expression $\SL_\lambda(\xb)-\SD_\lambda[\yb_1(\xb)]$ is an upper bound on the duality gap at $\xb$, as $\SD_\ml(\yb^*)-\SD_\ml[\yb_1(\xb)]=\SL_\ml(\xb^*)-\SD_\ml[\yb_1(\xb)]\leq \SL_\ml(\xb)-\SD_\ml[\yb_1(\xb)]$.
Direct calculation gives
\begin{align*}
 \SL_\ml(\xb)-\SD_\ml[\yb_1(\xb)] =  \frac{\|\vb\|_\infty^2}{\lambda} + \lambda\|\xb\|_1^2 - 2\,\xb\TT\vb \,,
\end{align*}
which gives the expression \eqref{B1} for $D_1$.
\end{proof}

\begin{coro}\label{C:B2}
Let $\xb\in\mathds{R}^{p}\setminus\{\0b\}$, and define
\begin{align}
D_2(\ab_i;\xb) &=\ D_0\Big(\ab_i;\ \yb_2,\ \lambda\, \cb\TT \Mb^{-1}(\widehat\wb(\xb))\cb-\SD_\lambda(\yb_2)\Big) \nonumber \\
 &=\
 \|\vb'\|_\infty - |v_i'| - \left( \frac{\|\vb'\|_\infty^2}{\lambda} + \yb_2\TT[\yb_2-\cb] \right)^{1/2}\sqrt{\|\ab_i\|^2+\ml} \,, \label{B2}
\end{align}
where $\yb_2=\yb_2(\widehat\wb(\xb))=\ml\,\Mb^{-1}(\widehat\wb(\xb))\cb$ and $\vb'=\vb'(\xb)=\Ab\TT \yb_2(\widehat\wb(\xb))$. Then, any $\ab_i$ satisfying $D_2 (\ab_i;\, \xb) >0$
cannot support an optimal design.
\end{coro}

\begin{proof}
This time, we use $\lambda\phi_c(\widehat \wb(\xb))-\SD_\lambda[\yb_2(\widehat\wb(\xb))]$ as an upper bound on the duality gap at $\xb$, which follows from
$\SD_\ml(\yb^*)-\SD_\lambda[\yb_2(\widehat\wb(\xb))]=\lambda\phi_c(\wb^*)-\SD_\lambda[\yb_2(\widehat\wb(\xb))]\leq \lambda \phi_c(\widehat \wb(\xb))-\SD_\lambda[\yb_2(\widehat\wb(\xb))]$.
Direct calculation gives
$\lambda\, \cb\TT \Mb^{-1}(\widehat\wb(\xb))\cb-\SD_\lambda[\yb_2(\widehat\wb(\xb))] = \|\vb'\|_\infty^2/\ml+ \yb_2\TT(\widehat\wb(\xb))[\yb_2(\widehat\wb(\xb))-\cb]$ and therefore \eqref{B2}.
\end{proof}

More generally, as $\lambda\phi_c(\widehat \wb(\xb))\leq\SL_\ml(\xb)$, see~\eqref{L>phi}, the substitution of $\lambda\phi_c(\widehat \wb(\xb))$ for $\SL_\ml(\xb)$ in the construction of $D_1(\ab_i; \xb)$ yields a better bound
\[
\widetilde D_1(\ab_i; \xb)= D_0\big(\ab_i; \yb_1(\xb), \lambda\phi_c(\widehat \wb(\xb))-\SD_\ml(\yb_1(\xb))\big) \geq D_1(\ab_i; \xb) \,.
\]
The reason for using this refined bound on the duality gap in $D_2$ (see Corollary~\ref{C:B2}) but not in $D_1$ is computational efficiency: the calculation of $D_1$ is dominated by the computation of $\vb$, which requires $O(mp)$
operations, while the calculation of $D_2$ is dominated by the computation of $\yb_2(\widehat\wb(\xb))$ that takes $O(m^2p)$ operations.
However, if we are willing to pay this higher computational cost then we might as well take advantage of the improved bound on the duality gap, hence its use in $D_2$.

In practice, the exact complexity of a screening operation depends on the algorithm used to solve the optimization problem~\eqref{qlasso}. The above discussion assumes that this algorithm works directly on the primal variables $\xb$, such as FISTA or the block coordinate descent proposed in~\citep{SagnolP2019}. An example comparing the efficiencies of screening with $D_1$ and $D_2$ is given in
Section~\ref{S:MNIST}.
However, the situation is quite different when the algorithm operates on the design weights $\wb$, such as the multiplicative update algorithm~\eqref{multiplicative}: the vector $\yb_2(\wb)$ is then readily obtained from
intermediate calculations when computing the gradient
$\nabla \phi'_c(\wb)$, and there is no computational incentive to use $D_1$ rather than $\widetilde{D}_1$ anymore. Moreover, in this case, starting from $\wb$, we use $\xb=\widehat\xb(\wb)$ in $\widetilde D_1$ and $\widehat\wb(\xb)=\wb$ in $D_2$, see~\eqref{hatw} and \eqref{widehat-x}. Lemma~\ref{L:y1=y2} then implies that  $\yb_2(\wb)=\yb_1(\xb)=\ml\Mb^{-1}(\wb)\cb$,
so that $\widetilde{D}_1$ and $D_2$ coincide. Rewriting this bound as a function of $\wb$, we obtain
\begin{align*}
D_2'(\ab_i; \wb) &= \|\vb\|_\infty- |v_i| - \left( \frac{\|\vb\|_\infty^2}{\lambda} + \yb\TT(\yb-\cb) \right)^{1/2}\, \sqrt{\|\ab_i\|^2+\ml} \,,
\end{align*}
where $\yb=\yb(\wb)=\ml\Mb^{-1}(\wb)\cb$ and $\vb=\Ab\TT\yb$.
We thus have the following property.

\begin{coro}\label{C:bound-c-opt}
Let $\Mb$ be any matrix in the convex hull of $\mathds{H}$ with the $\Hb_i$ satisfying \eqref{Hi}. Then, any matrix $\Hb_i$ in $\mathds{H}$ such that $B(\Mb,\Hb_i)>0$, where
\begin{align*}
B(\Mb,\Hb_i)= \left\{(1+\delta)\, \Phi_c(\Mb)-\ml\,\cb\TT\Mb^{-2}\cb\right\}^{1/2} - \left[\delta\, \Phi_c(\Mb)\left(1+ \frac{\|\ab_i\|^2}{\ml}\right)\right]^{1/2} - |\cb\TT\Mb^{-1}\ab_i|
\end{align*}
and $\delta=\max_{\Hb_i\in\mathds{H}} \cb\TT\Mb^{-1}\Hb_i\Mb^{-1}\cb/(\cb\TT\Mb^{-1}\cb)-1$, does not support a $c$-optimal design.
\end{coro}

\begin{proof}
We have $\Mb=\Mb(\wb)$ for some $\wb\in\SP_p$. Direct calculation gives $D_2'(\ab_i; \wb)=\ml\,B(\Mb,\Hb_i)$.
\end{proof}

Note that $\delta\geq 0$ and $\Phi_c(\Mb)\leq (1+\delta) \phi_c(\wb^*)$ with $\wb^*$ a $c$-optimal design, see \citet{PS2021-JSPI}.

\section{A homotopy algorithm for the quadratic Lasso}\label{S:homotopy}

In the standard Lasso problem~\eqref{Palpha} (with non-squared penalty $\alpha \|\xb\|_1$), the dual problem geometrically corresponds to the projection of the point
$\cb/\alpha$
onto the polytope $\mathcal{Q}=\{\zb: \|\Ab\TT \zb\|_\infty \leq 1\}$. This was used by~\citet{osborne2000new} to design a homotopy algorithm which starts from a large value of $\alpha$ (such that
$\cb/\alpha\in\mathcal{Q}$ and the projection is trivial) and stores the required data to maintain the projection of
$\cb/\alpha$
onto $\mathcal{Q}$ as $\alpha$ decreases.
One advantage of this technique is that it can be used to compute the full \emph{regularization path}, i.e., a set of optimal solutions $\xb^*(\alpha)$ to Problem~\eqref{Palpha} for all $\alpha>0$.
This is particularly interesting in applications of the Lasso algorithms for which the value of the regularization parameter $\alpha$ must be tuned using cross-validation, as in the least angle regression (LARS) method
for model selection~\citep{efron2004least}.
In this section, we
use Theorem~\ref{T:lasso-qlasso}
to adapt the homotopy method to the case of a quadratic Lasso penalty, and hence to the case of $c$-optimality.

In $c$-optimal design, $\lambda$ is a fixed constant
related to the variance of the observations and the variance-covariance matrix of the prior (see Section~\ref{sec:prelim}),
and computing the full regularization path $\big(\wb^*(\lambda)\big)_{\lambda\geq 0}$ is not essential. The example in
Section~\ref{S:MNIST}
nevertheless illustrates that the proposed homotopy method can yield a faster design construction than standard first-order algorithms, by several orders of magnitude, although the objective is to solve Problem~\eqref{phi} for a unique $\lambda>0$. On the negative side, one should notice that pathological cases are known for the standard lasso problem~\eqref{Palpha} for which the regularization path can have up to $3^n$ breakpoints~\citep{DBLP:conf/icml/MairalY12}.

\medskip
The main idea of the homotopy algorithm is that the projection $\zb^*(\alpha)$ of $\cb/\alpha$ over
$\mathcal{Q}$ has a closed form solution if we know on which face of $\mathcal{Q}$
the projection lies. By tracking the subset of equalities of the form
$\ab_i\TT\zb^*(\alpha)=\pm 1$ that hold as $\alpha$ is decreased, the algorithm
can detect the next ``breakpoint'' in the regularization path, that is, the largest value of $\alpha$ below the current value such that $\zb^*(\alpha)$ changes of face of $\mathcal{Q}$. Moreover, it can be seen that $\zb^*(\alpha)$ is
linear between two successive breakpoints, and we also obtain a corresponding
primal optimal solution $\xb^*(\alpha)$ that is continuous and piecewise linear with respect to $\alpha$. In the end, the algorithm produces a decreasing sequence
of breakpoints $\|\Ab\TT\cb\|_{\infty}=\alpha_1>\alpha_2>\ldots>\alpha_N=0$
and of corresponding solutions $\xb^*_1,\ldots,\xb^*_N$, such that
$\xb^*(\alpha)=\xb_1^*=\0b$ is optimal for all $\alpha\geq\alpha_1$, and
for all $k$ and  $\alpha$ of the form $\alpha=(1-\theta)\alpha_{k+1} + \theta \alpha_k\in[\alpha_{k+1},\alpha_k]$ for some $\theta\in[0,1]$, an optimal solution
is given by
\begin{align}\label{pwl}
 \xb^*(\alpha) = (1-\theta) \xb_{k+1}^* + \theta \xb_{k}^*.
\end{align}

The situation is in fact
similar for the quadratic lasso. Recall the geometric interpretation
of the dual problem~\eqref{conic_form}:
$\ybb^*=(\yb^{*\top},\ u^*)\TT$ is the projection of the vector
$\cbb=(\cb\TT,\ 0)\TT$ onto the polyhedral cone $\cP_\ml(\Ab)$ given by~\eqref{Plambda}. The projection $\ybb^*$ lies on a face of $\cP_\ml(\Ab)$, which can be identified with the subset of equalities of the form
$\ab_i\TT\yb^* = \pm \sqrt{\ml} u^*$ satisfied at $\ybb^*$.
The sequence of faces on which $\ybb^*$ lies as $\lambda$ decreases is
in one-to-one correspondence with the faces visited by the homotopy
algorithm for the standard lasso problem.
Rather than re-inventing the wheel by following the path
of $\ybb^*$ on the faces of $\cP_\ml(\Ab)$ as $\ml$ decreases, we are going
to use Theorem~\ref{T:lasso-qlasso} to reparametrize the regularization path
between two successive breakpoints as a function of $\lambda$.

\begin{theo}\label{T:adapt_lasso}
 Let $(\alpha_1,\xb_1^*),\ldots,(\alpha_N,\xb_N^*)$ with $\|\Ab\TT\cb\|_\infty=\alpha_1>\ldots>\alpha_N=0$ denote the sequence of breakpoints and
 corresponding solutions to the standard lasso problem returned by
 the homotopy algorithm. Define $\ml_1=\infty$ and $\lambda_k=\frac{\alpha_k}{\|\xb_k^*\|_1}$ for
 all $k\in\{2,\ldots,N\}$. Then, we have
 $\ml_1>\ml_2>\ldots>\ml_N=0$, and for all $\ml\in [\ml_{k+1},\ml_k)$,
 an optimal solution to the quadratic lasso~\eqref{qlasso} is given by
 \begin{equation}\label{xlambda}
  \xb^*=
  \frac{(\alpha_k-\ml\|\xb_k^*\|_1)\xb_{k+1}^* + (\lambda\|\xb^*_{k+1}\|_1-\alpha_{k+1}) \xb_k^*}
  {\alpha_{k} - \alpha_{k+1} + \lambda ( \|\xb^*_{k+1}\|_1-\|\xb^*_{k}\|_1)}.
 \end{equation}
 and a $c-$optimal design is
 \begin{equation}\label{wlambda}
 \wb^*=\widehat{\wb}(\xb^*) =
\frac{(\alpha_k-\ml\|\xb_k^*\|_1)\cdot |\xb_{k+1}^*| + (\lambda\|\xb^*_{k+1}\|_1-\alpha_{k+1})\cdot |\xb_k^*|}
  {\alpha_{k}\|\xb^*_{k+1}\|_1-\alpha_{k+1}\|\xb^*_{k}\|_1},
 \end{equation}
 where $|\xb_k^*|\in\mathds{R}^p$ represents the coordinate-wise absolute values of $\xb_k^*$, that is,  $(|\xb_k^*|)_i=|\xb_{k,i}^*|, \forall i\in[p]$.
 In particular, $\ml\mapsto \wb^*$ is continuous and piecewise linear.
\end{theo}
\begin{proof}
We first observe that
 $0=\|\xb_1^*\|\leq \|\xb_2^*\|\leq \ldots\leq\|\xb_N^*\|$, which
 follows from the convexity of the Pareto-front for the minimization
 of the two objectives $\|\Ab\xb-\cb\|^2$ and $\|\xb\|_1$;
 see~\citet[\S 4.7.5]{BoydV2004}.
  Thus, we have $\lambda_1>\ldots>\lambda_N$, and for all $\lambda>0$ there exists
  $k\in\{1,\ldots,N-1\}$ such
 that $\lambda_{k+1}\leq \lambda< \ml_k$. By Theorem~\ref{T:lasso-qlasso}
 and~\eqref{pwl}, for all $\theta\in[0,1]$ the point $\xb^*\big((1-\theta)\alpha_{k+1}^*+\theta \alpha_k^*\big)=(1-\theta)\xb_{k+1}^* + \theta \xb_k^*$ solves the
 quadratic lasso with regularizer
 \[\ml=\frac{(1-\theta)\alpha_{k+1}^*+\theta \alpha_k^*}{\|(1-\theta)\xb_{k+1}^* + \theta \xb_k^*\|_1}
 =
 \frac{(1-\theta)\alpha_{k+1}^*+\theta \alpha_k^*}{(1-\theta)\|\xb_{k+1}^*\|_1 + \theta \|\xb_k^*\|_1},
  \]
 where the second equality comes from the fact that no coordinate of $\xb^*(\alpha)$ changes sign between two consecutive breakpoints, that is, $x_{k,i}^*\cdot x_{k+1,i}^* \geq 0$ holds for all $i\in[p]$.
 Solving for $\theta$ yields
 $\theta=\frac{\lambda\|\xb^*_{k+1}\|_1-\alpha_{k+1}}{
 \alpha_{k} - \alpha_{k+1} + \lambda ( \|\xb^*_{k+1}\|_1-\|\xb^*_{k}\|_1)}.$
 It is easy to verify that $\theta$ lies in $[0,1]$, as desired, by using
 the inequalities
 $\alpha_{k+1} \leq \ml \|\xb_{k+1}^*\|_1$,
 $\ml \|\xb_{k}^*\|_1 \leq \alpha_k$,
 $\alpha_k>\alpha_{k+1}$ and $\|\xb_{k+1}^*\|_1\geq \|\xb_k^*\|_1$. Substituting
 the value of $\theta$ in
 $\xb^*=(1-\theta)\xb_{k+1}^* + \theta \xb_k^*$ yields the formula~\eqref{xlambda}.
 Finally, we replace $\xb^*$ by its value in $\wb^*=\widehat{\wb}(\xb^*)$ to obtain~\eqref{wlambda}.
\end{proof}

The steps of the homotopy algorithm for the standard lasso problem are summarized
in lines~1--\ref{l12} of Algorithm~\ref{alg:Halpha}, following the description
of the algorithm given by~\citet{DBLP:conf/icml/MairalY12}.
Here, the notation $\Ab_J=[\ab_j]_{j\in J}\in\mathds{R}^{m \times |J|}$ stands for the submatrix of $\Ab$ formed by the columns indexed in $J$,
and $\bar{J}:=[p]\setminus J$ denotes the complement of $J$ in $[p]$.
Note that practical enhancements can easily be added. In particular,
one can replace lines~\ref{l8}--\ref{l11} by a closed-form expression to find the next break point $\alpha$. One can also
store the inverse matrix $(\Ab_{J}\TT\Ab_{J})^{-1}$
to accelerate the computation of
the vectors $(\Ab_J\TT\Ab_J)^{-1}\Ab_J\TT\cb$ and $(\Ab_J\TT\Ab_J)^{-1}\mveb$,
and use rank-one update formulas at the end of each iteration;
see~\citet{osborne2000new} for more details.

Since by Theorem~\ref{T:adapt_lasso} we have
$\infty=\lambda_1>\lambda_2>\ldots>\lambda_N=0$, the algorithm
exits the while loop with $\ml \in [\ml_{k+1},\ml_k)$, so the
vector $\xb^*$ (or $\wb^*$) returned by Algorithm~\ref{alg:Halpha}
is an optimal solution to the quadratic lasso
(respectively, a $c$-optimal design).

\begin{algorithm}[t]
\caption{Homotopy algorithm for the quadratic lasso/$c$-optimal design.
}\label{alg:Halpha}
\begin{algorithmic}[1]
\Procedure{Homotopy}{$\Ab=[\ab_1,\ldots,\ab_n]\in\mathds{R}^{m\times n},\cb\in\mathds{R}^m,\lambda>0$}
\State {\textbf{Initialization:}} $\alpha_1 \gets \|\Ab\TT\cb\|_\infty$;\quad $\xb_1^*\gets \0b$;\quad $\ml_1\gets \infty$
\State \phantom{\textbf{Initialization:}} $j_0\gets \arg\max_{j\in[p]} |\ab_j\TT\cb|$;\quad $J\gets\{j_0\}$;\quad $k\gets 0$
\While{$\ml<\ml_{k+1}$}
    \State $k\gets k+1$
    \State $\mveb\gets \operatorname{sign}(\Ab\TT(\cb-\Ab\xb_k^*))\in\{-1,0,1\}^p$ %$$ $
    \State For all $\alpha\in\mathds{R}$ define $\xib(\alpha)\in\mathds{R}^p$ by $\left\{\begin{array}{l}
                                                 \xib_J(\alpha) = (\Ab_J\TT\Ab_J)^{-1}(\Ab_J\TT\cb-\alpha\cdot \mveb_J);\\
                                                 \xib_{\bar J}(\alpha) = \0b.
                                                \end{array}
                                                \right.$
    \State \label{l8} Find the largest $\alpha\in[0,\alpha_k)$ such that either:
    \State \quad $\bullet$
    $|\ab_j\TT(\cb-\Ab_J\xib_J(\alpha))|=\alpha$ for some $j\in \bar{J}$; In that case
    set $J\gets J\cup\{j\}$
    \State   \quad $\bullet$
    $\xib^*_j(\alpha)=0$ for some $j\in J$; In that case, set $J\gets J\setminus\{j\}$
    \smallskip
    \State \label{l11} If no such $\alpha$ exists, set $\alpha\gets 0$
    \State $\alpha_{k+1}\gets \alpha$;\quad $\xb_{k+1}^*\gets \xib^*(\alpha_{k+1})$;\quad $\lambda_{k+1}\gets \frac{\alpha_{k+1}}{\|\xb_{k+1}^*\|_1}$
\EndWhile \label{l12}
\smallskip
\smallskip
 \State \Return $\xb^*$ given by~\eqref{xlambda}  and/or the $c$-optimal design $\wb^*$ given by~\eqref{wlambda}.
\EndProcedure
\end{algorithmic}
\end{algorithm}

\begin{remark} \label{remark:degeneracy}
Degeneracy can occur if the polytope $\mathcal{Q}$ has a particular geometrical structure, which can prevent Algorithm~\ref{alg:Halpha} from terminating
(this might happen if the largest $\alpha$ we are looking for at line~\ref{l8} is realized for several distinct indices $j\in J$ or $j\in \bar{J}$).
To avoid this, a typical approach consists in adding random noise to the data, so that the problem becomes non-degenerate with probability $1$. Alternatively, we can choose the index entering or leaving $J$ arbitrarily among the maximizers,
and implement an anti-cycling rule to ensure that the same face of $\mathcal{Q}$ is not visited twice; see, e.g., \citet{gill1989practical}. One should note that, even for non-degenerate problems, numerical issues can occur when the decrease of $\alpha$ at a given iteration is less than machine precision.
\fin
\end{remark}

\begin{remark}
The generalization of this approach to the quadratic group lasso (and thus to $L$-optimality) is not straightforward,
as the polyhedral structure of the dual feasible region is lost when $r>1$;
see~\eqref{dualL}.
Nevertheless we point out that there have been attempts
to design homotopy algorithms for more general
penalty functions~\citep{zhou2014generic}, such as group-lasso~\citep{yau2017lars},
relying on the fact that the optimal solution $\Xb^*(\alpha)$ solves
an ordinary differential equation in the interior of the segment $(\alpha_1,\alpha_2)$
between two breakpoints of the regularization path.
\end{remark}

 %_____________________________________________________________________________________________________________
\section{Examples}\label{S:examples}

The example of Section~\ref{S:MNIST} is directly taken from the literature on Lasso regression. A second example, concerning $L$-optimal design, is presented in
Section~\ref{S:IMSE}.

\subsection{Optimal design for image classification}\label{S:MNIST}

We use the training set of the well known MNIST database~\citep{lecun-mnisthandwrittendigit-2010},
which contains 60\,000 28$\times$28 images
in grey levels of handwritten digits, each associated with a label for the digit in $\{0,\ldots,9\}$ it represents.
For our purpose, we build a smaller training set of $600$ images of each digit, which we arrange in a matrix $\Ab \in \mathds{R}^{784 \times 6\,000}$. The $i$th column $\ab_i$ is a vector
representing the levels of grey of the $m=28^2=784$ pixels of one image, and has been normalized so that $\|\ab_i\|=1$. As there are 600 images for each digit in $\{0,\ldots,9\}$, $p=6\,000$. We choose an arbitrary image that does not belong to the training set, and transform it as above into a vector $\cb \in \mathds{R}^{784}$ of unit norm.
The $c$-optimal design problem
thus attempts to find which images should be labeled to best predict the label of the image represented by the vector $\cb$, and one can legitimately assume that these images will represent the same digit.

This problem illustrates well the equivalence between a design problem (where we try to select samples that ``span'' the subspace $\cb\cdot \mathds{R}$ well) and a Lasso problem (where the data matrix has been transposed, so
samples become features and vice-versa, and we want to select features that ``explain'' the target vector $\cb$ well).
For a given value of $\lambda$ penalizing non-sparse solutions, the Lasso problem~\eqref{qlasso} tries to express the vector $\cb$ as a linear combination of a few images $\ab_i$. If the label of $\cb$ is unknown, this can be interpreted as a task of supervised learning: for well chosen values of the penalty parameter $\lambda$, we expect that the vast majority of $\ab_i$'s supporting the solution $\xb^*$ will have the same label as $\cb$, which can be used to classify the unknown image. Alternatively, if no labels were provided, we could use this approach repeatedly in a leave-one-out manner (i.e., select one sample to be the vector $\cb$  and use the other $p-1$ samples to build the matrix $\Ab$) to perform an unsupervised learning task and construct clusters of images that are likely to share the same label.

The 10 support points of the optimal solution for $\lambda=0.4$ are displayed in Figure~\ref{fig:homotopy_sixes} (left), together with the image corresponding to the vector~$\cb$: as expected, all images represent the same digit (here, a six).
Figure~\ref{fig:lasso_D1D2} shows the evolution of
the duality gap $\SL_\lambda(\xb)-\SD_\lambda(\yb_1(\xb))$
for $\lambda=0.4$,
when the Coordinate Descent (CD) algorithm\footnote{\citet{SagnolP2019} actually present a
\emph{block} CD
algorithm for Bayes $L$-optimality. For $c$-optimal designs, the blocks are of size $r=1$ so this is indeed a CD algorithm.}
described in~\citep{SagnolP2019} is used
(left), as well as the percentage of support points $\rho(k)$ that have been
eliminated after $k$ iterations (right). As $m$ is large
and $r=1$ (this is a $c$-optimal design problem), one iteration of the CD algorithm is much faster than one iteration of the multiplicative algorithm~\eqref{multiplicative}:
the former has complexity $\mathcal{O}(m\cdot r)$ to update a
single row of the matrix $\Xb$; when $r=1$, this means that a single coordinate of $\xb$ is updated in $\mathcal{O}(m)$; an iteration cycling over $p_i=(1-\rho(i))\cdot p$ remaining coordinates has thus complexity $\mathcal{O}(p_i\cdot m)$. In contrast, the multiplicative algorithm is dominated by the computation of the information matrix~$\Mb(\wb)$, which takes $\mathcal{O}(p_i\cdot m^2)$ when $p_i$ coordinates remain.
See Table~\ref{tab:runtime} and its presentation below for an empirical comparison of the running times\footnote{Calculations are in python on a PC at 2.7 GHz and 16 GB RAM.} of the different algorithms.

We only compare the bounds $D_1$ and $D_2$, as for large $m$
the overhead of using $B_1$, $B_2$ or $B_3$ is too high and $B_4$ cannot be applied for $r<m$.
Several observations can be made. First, when screening is performed every $5$ or $25$ iterations, on the right panel $\rho(k)$ is larger for $D_2$ than for $D_1$, with the consequence that $D_2$ eliminates points slightly earlier than $D_1$ (we observe no difference when screening is performed every 100 iterations, the two curves superimpose perfectly).
Second, the curves on the left panel show that the computational cost of the screening test applied to one $\ab_i$ is generally higher for $D_2$ than for $D_1$.
Indeed, the rate of convergence is related to the slope of the curve showing the evolution of the duality gap. Ignore the initial part of the plot, where the number of remaining points to be tested varies much according to the method and periodicity. Once enough points have been eliminated, the slope for $D_2$ decreases when the period between two successive screenings decreases. On the opposite, the slope stays roughly constant for $D_1$: screening with $D_1$ every $5$, $25$ or $100$ iterations yields the same convergence rate; the more frequent the screening, the earlier the elimination, and thus the induced acceleration.

\begin{figure}[t]
 \includegraphics[width=0.52\textwidth]{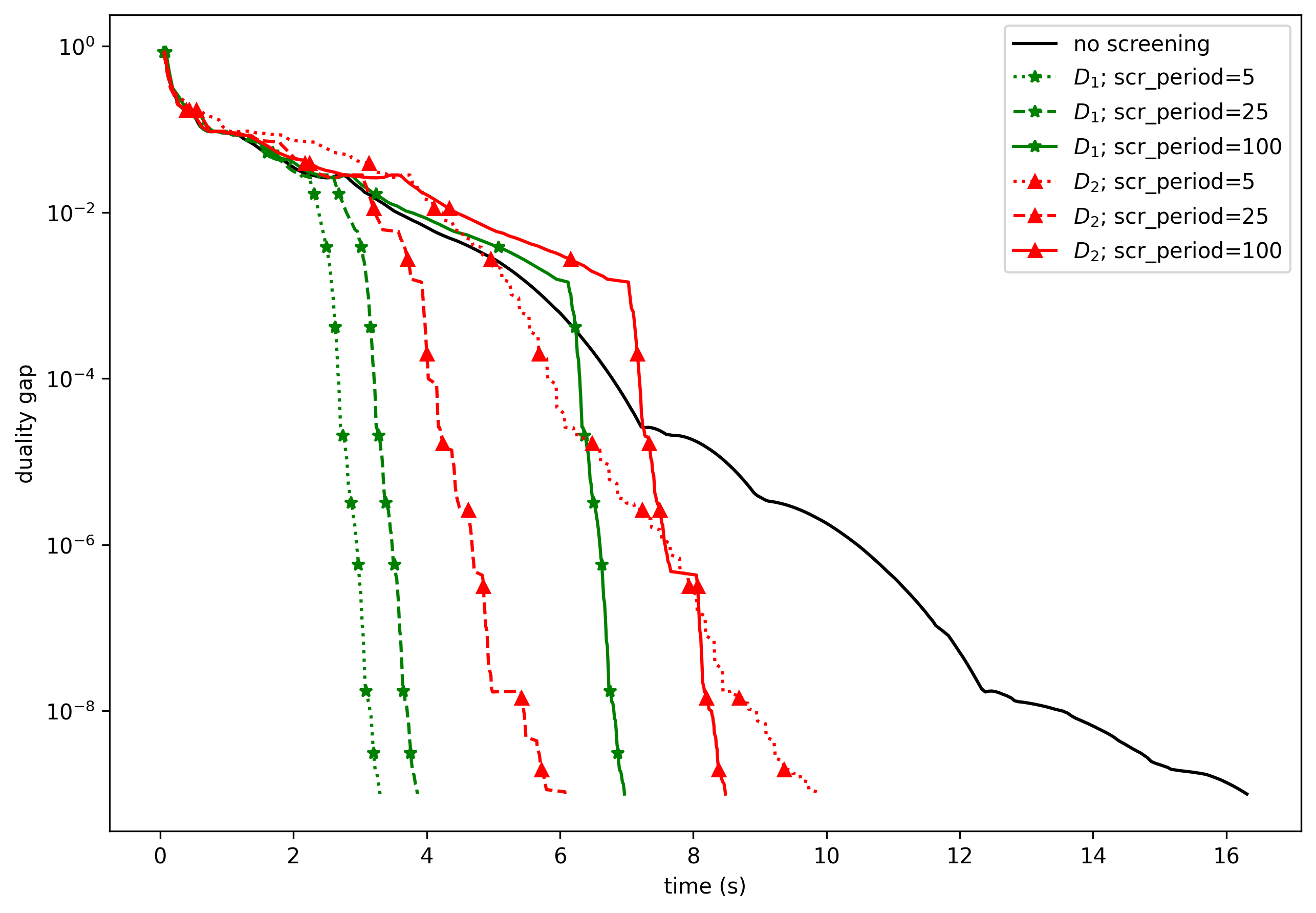}
 \includegraphics[width=0.48\textwidth]
 {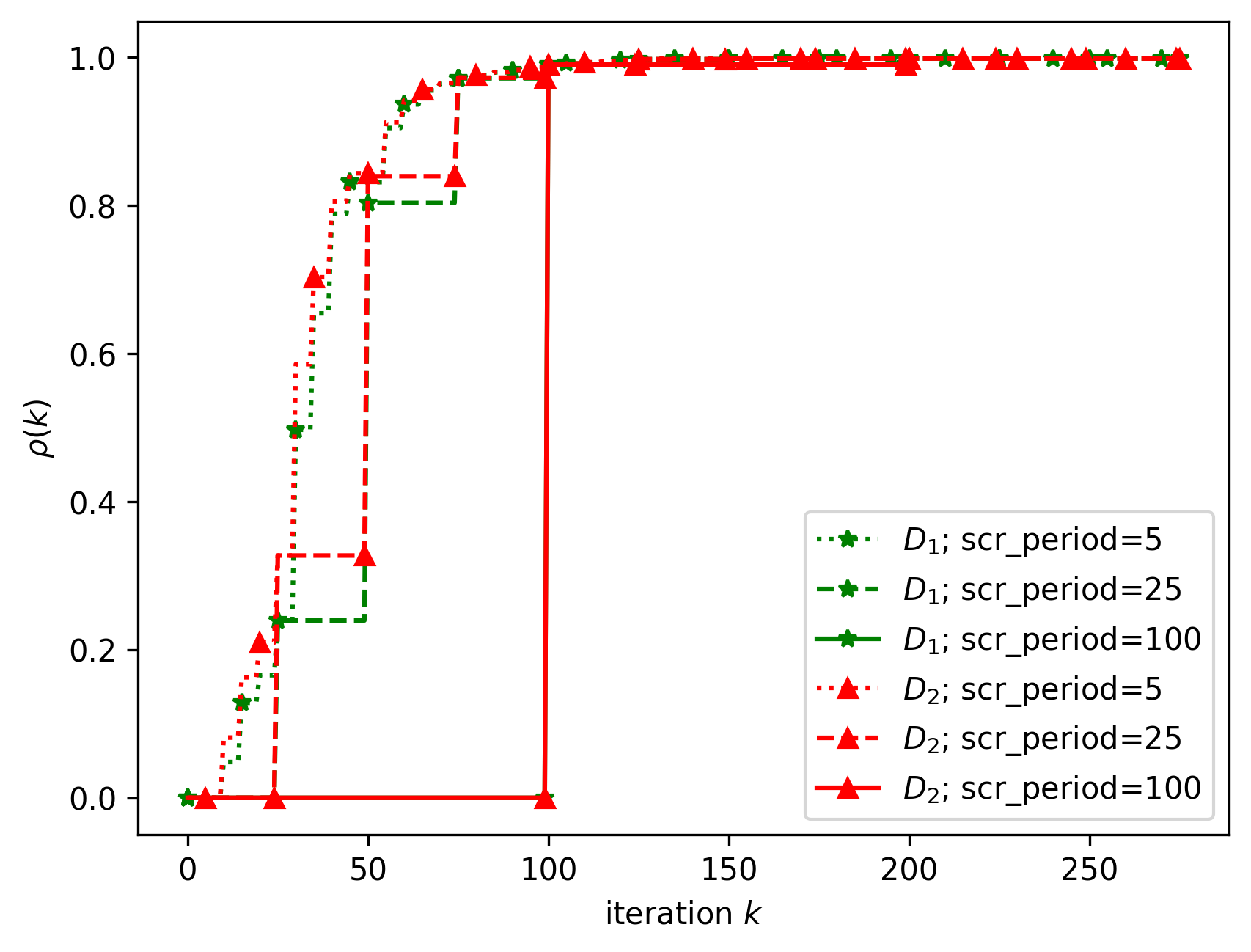}

 \caption{\small Evolution of the duality gap $\SL_\lambda(\xb)-\SD_\lambda(\yb_1(\xb))$ (left) and of the proportion $\rho(k)$ of points eliminated after $k$ iterations (right) for the screening bounds $D_1$ and $D_2$ applied with different periodicities.
 \label{fig:lasso_D1D2}
 }
 \end{figure}

The effect of the regularization parameter $\lambda$
is shown on Figure~\ref{fig:homotopy_sixes} (right),
together with the convergence of the homotopy algorithm.
The problem is badly conditioned for small $\lambda$, which slows down convergence and delays elimination of non-supporting points (we use $D_1$ every 10 iterations).
Note, however, that the problem of $c-$optimal design obtained in the limit $\lambda\to 0$ can be solved efficiently by linear programming~\citep{harman2008computing}.
Interestingly, we see that the homotopy algorithm usually requires less time to get an \emph{exact} solution than the  CD algorithm needs to get a \emph{reasonable approximate} solution, for most values of $\lambda$. The speed-up can attain several orders of magnitude: for $\lambda=10^{-2}$, the homotopy algorithm computes an \emph{exact} solution to the problem in less than $4$~s, while the duality gap for the CD algorithm with $D_1$ is still $10^{-3}$ after $67$~s; for $\lambda=10^{-4}$, the homotopy algorithm finds an exact solution within $45$~s, while the CD algorithm with $D_1$ still has a duality gap or more than $10^{-4}$ after $2\,000$~s.

For the above experiment, the CD algorithm was used as a
reference because it performed better than
first-order algorithms. Table~\ref{tab:runtime} compares the running time of the homotopy algorithm and CD with several first order algorithms (Multiplicative weight updates (MWU), FISTA, Frank-Wolfe (FW)), as well the time required by
a commercial solver (Gurobi~\citep{gurobi} used with the PICOS interface~\citep{PICOS} in Python) to solve two different Second-Order Cone Programming (SOCP) formulations of the problem. The first formulation (SOCP-1) comes
from~\citep{Sagnol2011} and read $\min\{\cb\TT\xb: \|\xb\|^2\leq v,\ (\ab_i\TT\xb)^2\leq u_i,\ u_i+\lambda v\leq 1,\ \forall i\in[p]\}$. The second one
(SOCP-2) is a direct reformulation of~\eqref{qlasso}: $\min\{u+\lambda t: \|\Ab\xb-\cb\|^2\leq u, \|\xb\|_1 \leq v,\ v^2\leq t\}$, and performs
significantly better.
Note that the reported times are for the computation of an exact solution (up to machine precision) of the problem by using the homotopy method, while we have
used a tolerance of $10^{-4}$ for all other algorithms.
We used the screening rule $D_1$ with $\tau=10$ for CD, MWU, FISTA and FW. Observe that although the CD iterations are faster than the MWU iterations, this algorithm has a slow convergence for small values of $\lambda$. Thus, more iterations are required and MWU is faster than CD for $\lambda=10^{-4}$. Also,
due to the very slow convergence of FW and FISTA for small values of $\lambda$,
we have not been able to reach the desired tolerance within 1 hour in some
of the experiments. In contrast, the SOCP solver appears to be insensitive to
the value of $\lambda$.

\begin{table}[t]
\caption{\small Time (in s) to reach an exact
solution with the Homotopy algorithm and a
near-optimal solution (tolerance=$10^{-4}$) with different algorithms, for several values of $\lambda$. \vspace{-1.5em}
\label{tab:runtime}}
\begin{center}
 \begin{tabular}{crrrrr}
 \hline
  $\lambda$ & $10^0$ & $10^{-1}$ & $10^{-2}$ & $10^{-3}$ & $10^{-4}$\\\hline
  Homotopy  & 0.22   &   0.73    &   3.16    &   13.05   &   44.66\\
  CD        & 1.68   &   7.15    &   29.00   &   201.23  &   2012.15\\
  MWU       & 53.92  &  67.39    &  204.67   &   432.39  &   619.09\\
  FISTA     & 141.86 &  674.51   &  2011.01  &    --~~   &   --~~ \\
  FW        & 14.86  &   --~~    &  --~~     &    --~~   &   --~~ \\
  SOCP-1    & 62.57  & 66.41     & 72.59     & 72.16     & 64.86\\
  SOCP-2    & 13.54  & 13.71     & 13.53     & 13.81     & 13.91\\
  \hline
 \end{tabular}
 \end{center}
\end{table}

\begin{figure}[t]
\begin{minipage}{0.45\textwidth}
\vspace{-6mm}
 \includegraphics[width=\textwidth]{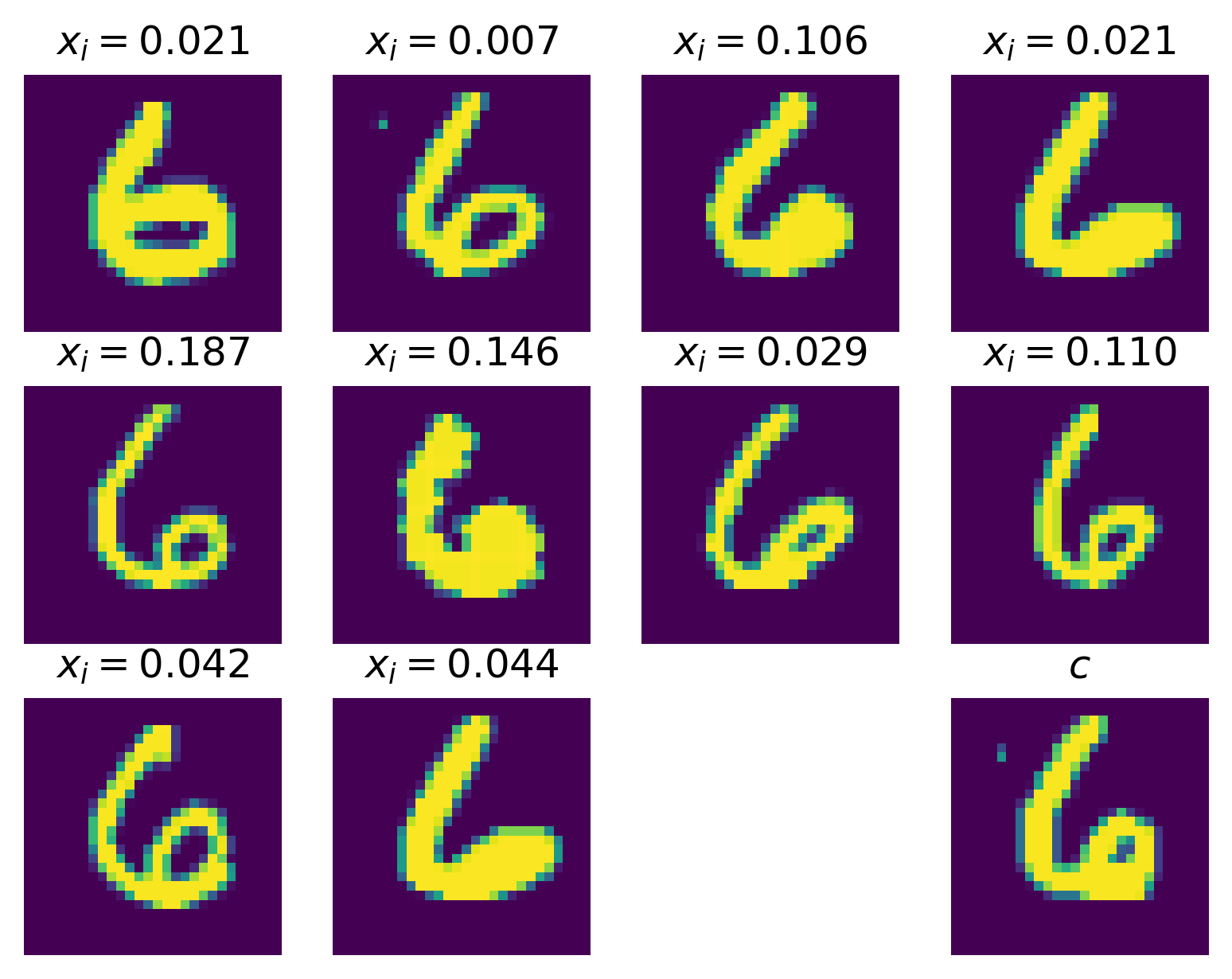}
\end{minipage}
\begin{minipage}{0.55\textwidth}
\includegraphics[width=\textwidth]{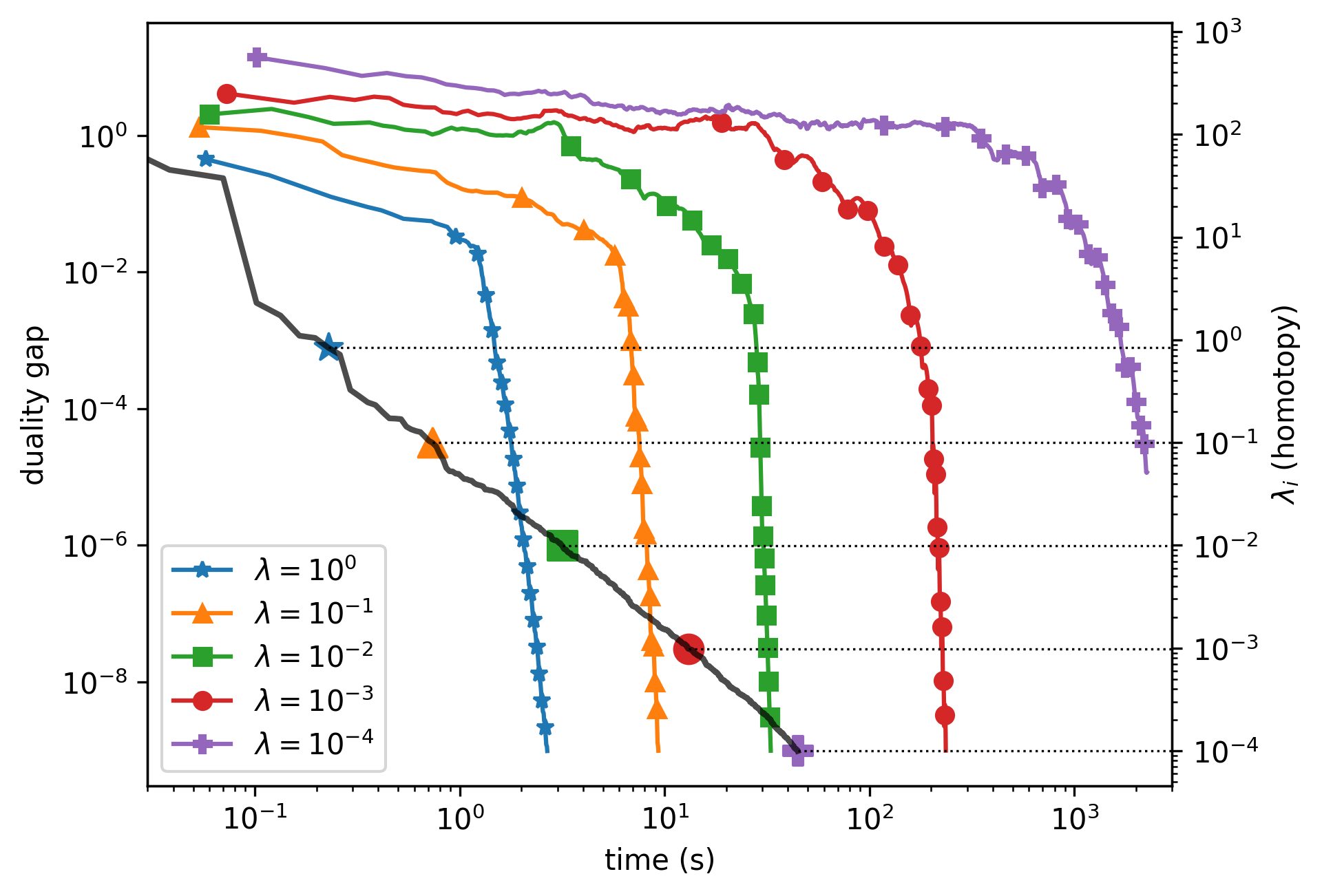}
\end{minipage} \vspace{-4mm}
\caption{\small
Left: Image corresponding to $\cb$ and support points of the optimal solution for $\lambda=0.4$. Right:
The colored curves indicate the convergence of the CD algorithm for different values of the regularization parameter $\lambda$ (duality gap shown on leftmost $y$-axis),
when $D_1$ is used every 10 iterations.
The black curve shows the progression of $\lambda_i$ (rightmost y-axis) along time for the homotopy algorithm. The markers on this curve indicate the points in time where the homotopy
algorithm has reached an \emph{exact} solution for the
corresponding value of $\lambda$.
\label{fig:homotopy_sixes} \vspace{-1em}
}
\end{figure}

\bigskip
Finally, in Figure~\ref{MNIST-L} we compare the effect of the screening rule $D_1$ on three algorithms for an $L$-optimal design problem with $r=50$. The training set contains $p=1\,200$ images (there are $120$ samples for each digit), the images have been down-sampled to $m=20^2=400$ pixels, and the matrix $\Kb\in\mathds{R}^{m\times r}$ is formed by a collection of $5$ images of each digit, randomly chosen outside the training set. The top-left panel of the figure shows the support of the corresponding $L$-optimal design (for $\lambda=0.4$): it contains at least one representative of each digit, and thus fulfills the objective of selecting a subset of samples that spans the subspace of all handwritten digits well. The other three panels show that application of the screening rule speeds up computations for the three algorithms considered, with, however, different gains. In this example, the multiplicative weight update algorithm benefits the most from the screening rule, but block CD yields the fastest convergence, with an acceleration factor of about two when $D_1$ is used every 10 iterations.

\begin{figure}[ht!]
\begin{minipage}[b]{.5\textwidth}
\includegraphics[width=\textwidth, height=0.85\textwidth]{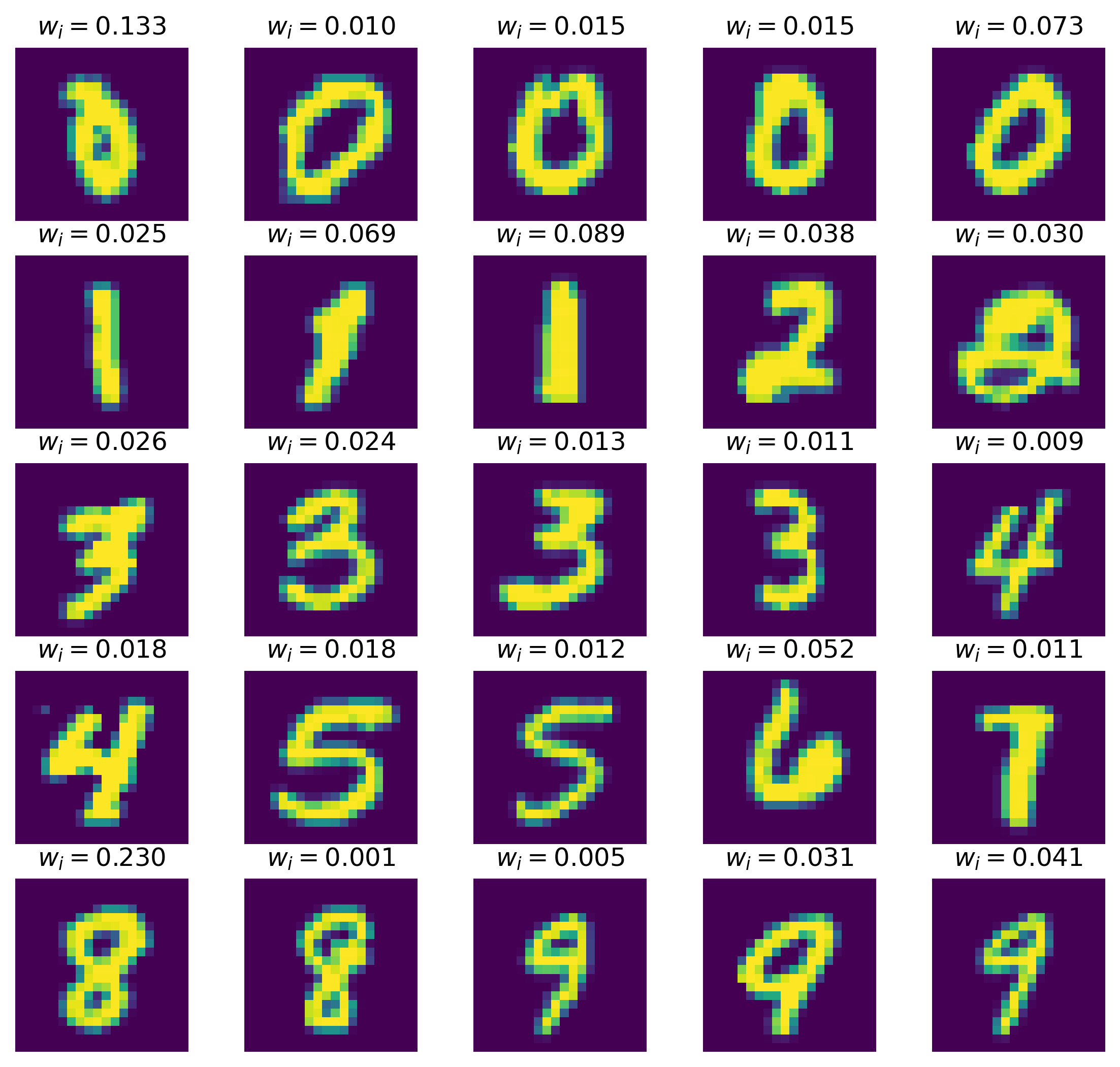}

\vspace{-1em}
\begin{center} (a) \end{center}

\includegraphics[width=\textwidth]{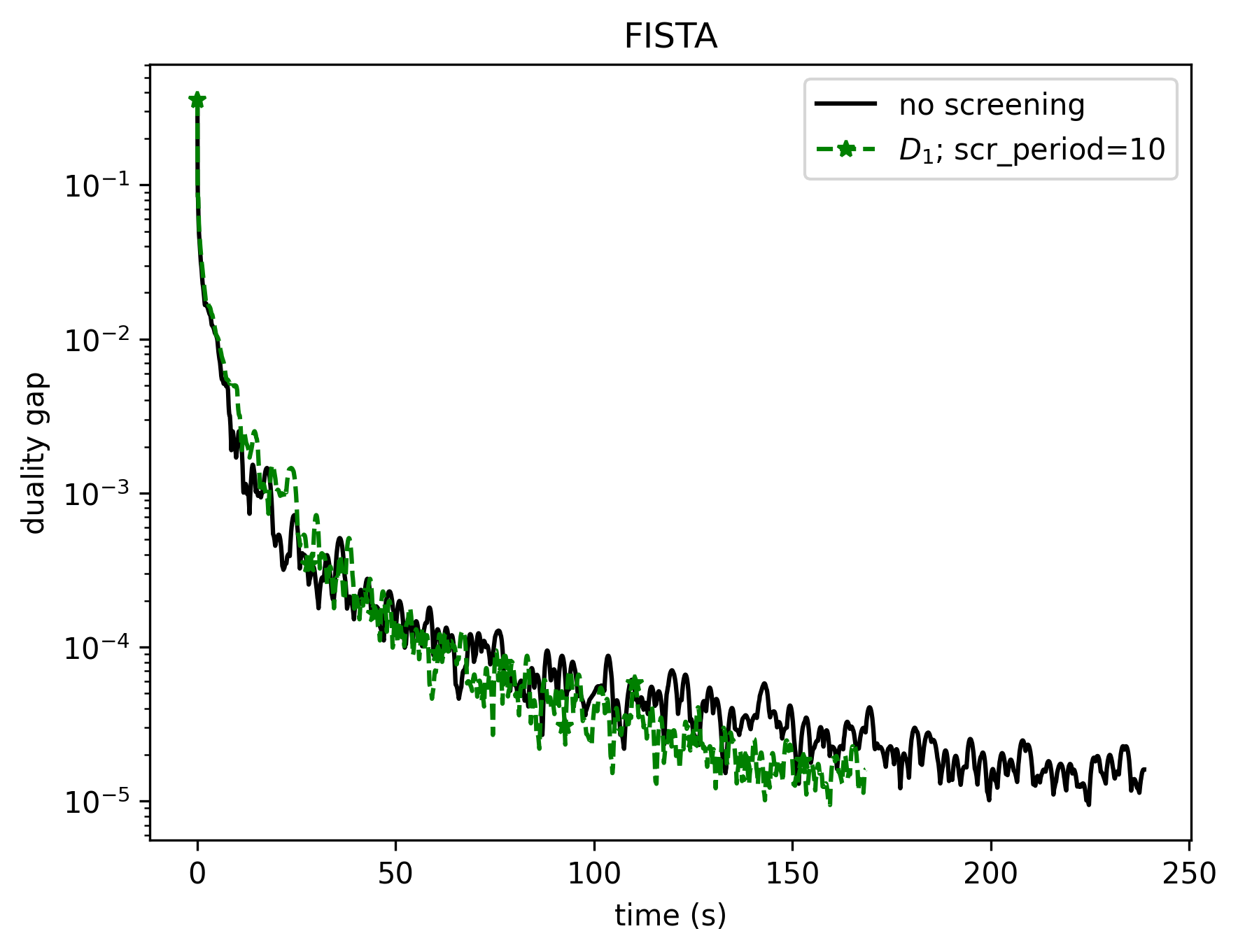}

\vspace{-1em}
\begin{center} (c) \end{center}
\end{minipage}
\begin{minipage}[b]{.5\textwidth}
\vspace{1em}
\includegraphics[width=\textwidth]{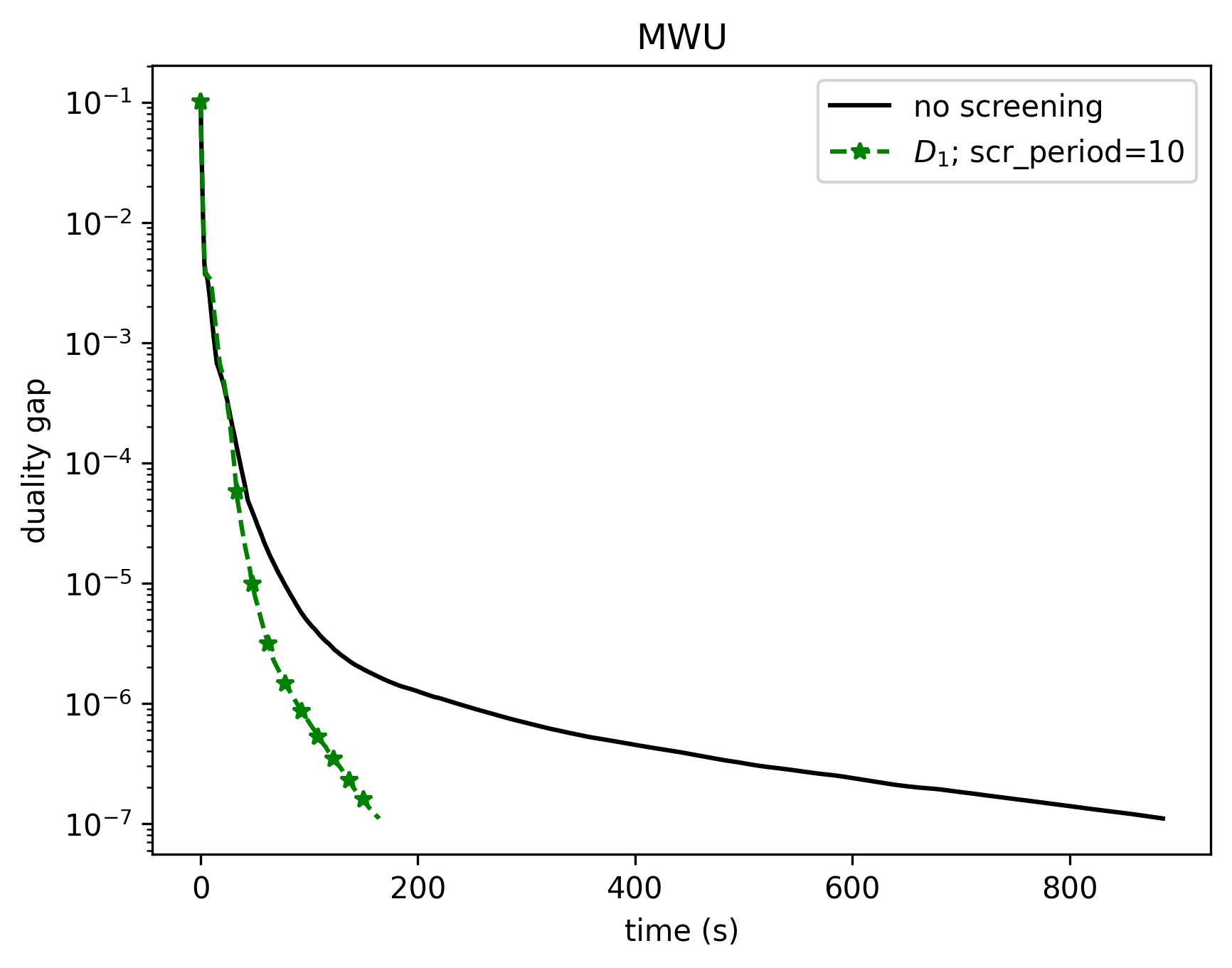}

\vspace{-1em}
\begin{center} (b) \end{center}
\includegraphics[width=\textwidth]{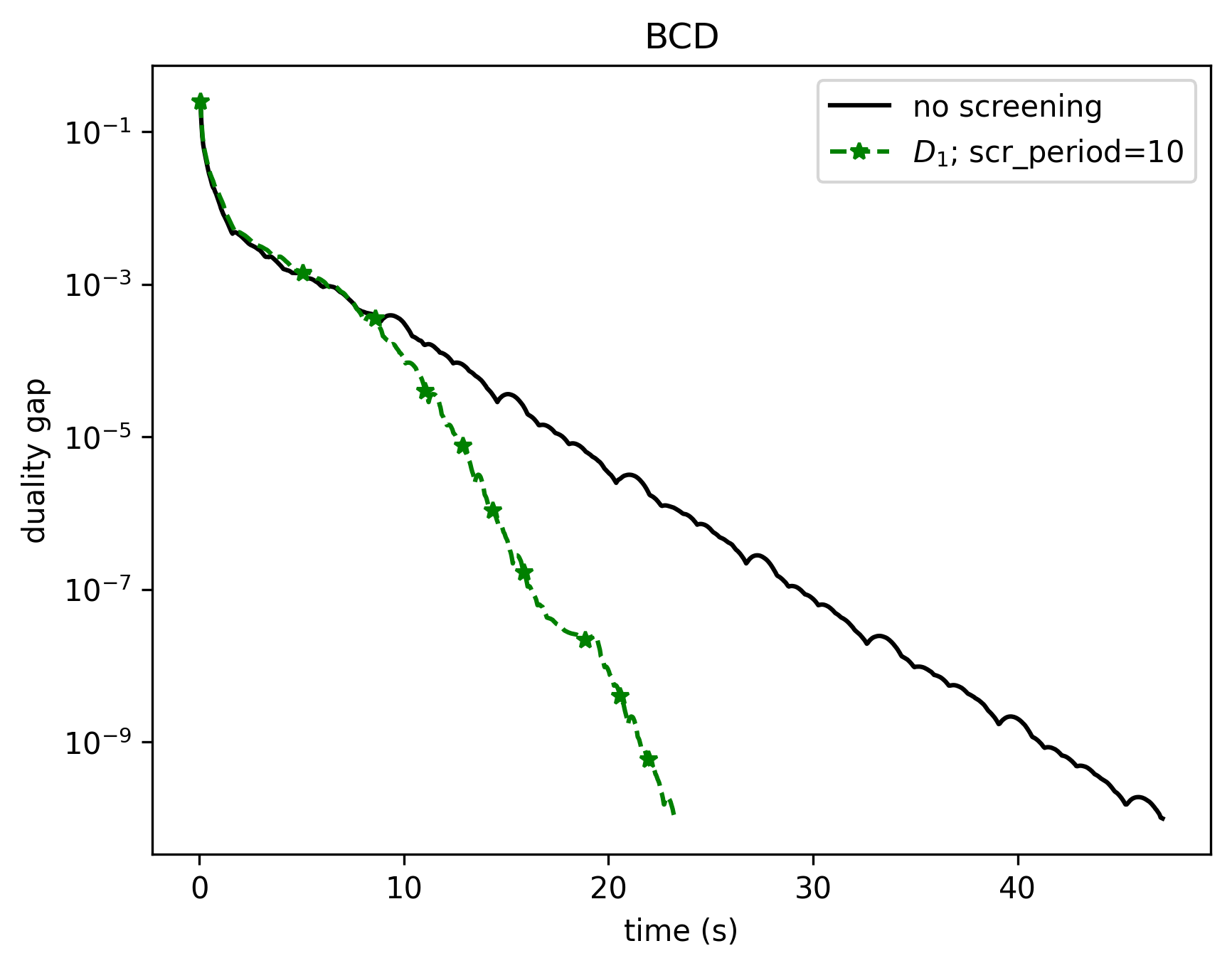}

\vspace{-1em}
\begin{center} (d) \end{center}
\end{minipage}

\caption{\small (a) $L$-optimal design for $\lambda=0.4$ over the MNIST database, with a training set of $p=1200$ images (120 per digit) of size $m=400=20\times 20$,
with $\Kb$ containing $r=50$ images, i.e., $5$ images of each digit (top left). (b-d) Convergence of three algorithms,
with and without the screening test $D_1$ run every 10 iterations: Multiplicative Weight Update (b); FISTA (c);
Block CD (d).  \label{MNIST-L}}
\end{figure}

%\subsection{An example of \texorpdfstring{$L$}{L}-optimal design: IMSE optimal design in random-field models}\label{S:IMSE}
\subsection{An example of $L$-optimal design: IMSE optimal design in random-field models}\label{S:IMSE}

{\sloppy
We consider the framework of \citep{GP-CSDA2016} for the minimization of the Integrated Mean-Squared Error (IMSE) in a random-field interpolation model $Z_\ssb$, indexed by a space variable $\ssb\in\SSS\subset\mathds{R}^d$. We assume that $Z_\ssb$ is Gaussian, centered, with covariance $\Ex\{Z_{\ssb_1}Z_{\ssb_2}\}=\mg(\ssb_1,\ssb_2)$. After observation at $n$ design points $\Sb_n=\{\ssb_1,\ldots,\ssb_n\}$ in $\SSS$, the IMSE for a given measure of interest $\mu$ on $\SSS$ is
\[\IMSE(\Sb_n)=\int_\SSS \left[\mg(\ssb,\ssb)-\mgb\TT_n(\ssb)\Gammab_n^{-1}\mgb_n(\ssb)\right]\, \dd\mu(\ssb),\]
with $\mgb_n(\ssb)=(\mg(\ssb_1,\ssb),\ldots,\mg(\ssb_n,\ssb))\TT$ and $(\Gammab_n)_{i,j}=\mg(\ssb_i,\ssb_j)$, $i,j=1,\ldots,n$.
}

The truncation of a particular Karhunen-Lo\`eve expansion of $Z_\ssb$ yields a Bayesian linear model, with the eigenfunctions as regression functions, and the construction of an IMSE-optimal design can be cast as an $L$-optimal design problem. The solution is much facilitated when $\mu$ is replaced by a discrete measure $\mu_p$ supported on a finite set $\SSS_p\subset\SSS$ and design points are selected within $\SSS_p$. Denote $\Db_p=\diag\{\mu_p(\ssb_1),\ldots,\mu_p(\ssb_p)\}$. We diagonalize the $p\times p$ matrix $\Db_p^{1/2}\Gammab_p\Db_p^{1/2}$ into $\Ub_p\Lambdab_p\Ub_p\TT$, with $\Lambdab_p=\diag\{\ml_1,\ldots,\ml_p\}$ and $\Ub_p\Ub_p\TT =\Ib_p$, and define $\Vb_p=\Db_p^{-1/2}\Ub_p$. The $i$th column $\vb_i$ of $\Vb_p$ corresponds to an eigenvector of $\Gammab_p\Db_p$
for the eigenvalue $\ml_i$ and satisfies $\vb_i\TT\Db_p\vb_j=\delta_{i,j}$ (the Kronecker symbol) for all $i,j$. We assume that the $\ml_i$ are ordered by decreasing values, and introduce a spectral truncation by using the $m$ largest eigenvalues only. Denoting by $\Phib_m$ the $p\times m$ matrix with $i$th column equal to $\vb_i$, $i=1,\ldots, m$, we obtain the Bayesian linear model
$Z_{\ssb_i}=\phib_{m,i}\TT \betab + \mve_{\ssb_i}$, where $\betab$ has the normal prior $\SN(\0b_m,\Lambdab_m)$, $\phib_{m,i}\TT$ is the $i$th row of $\Phib_m$, and where the $\mve_{\ssb_i}$ are normal random variables, independent of $\betab$, with zero mean and covariance $\Ex\{\mve_{\ssb_i},\mve_{\ssb_j}\}=\mg(\ssb_i,\ssb_j)-\phib_{m,i}\TT\Lambdab_m\phib_{m,j}$ for all $i,j$. Finally, we neglect correlations and approximate the IMSE via the estimation of $\betab$ by its posterior mean in
\begin{align}\label{BLM2}
Z'_{\ssb_i}=\phib_{m,i}\TT \betab + \mve'_{\ssb_i} \,,
\end{align}
where the errors $\mve'_{\ssb_i}$ are uncorrelated but heteroscedastic, with variance
$\ms_i^2=\Ex\{\mve_{\ssb_i},\mve_{\ssb_i}\}=\mg(\ssb_i,\ssb_i)-\phib_{m,i}\TT\Lambdab_m\phib_{m,i}$;
see \citet{GP-CSDA2016}.
After observations at $n$ points $\ssb_i$ with indices in $J_n\subseteq\{1,\ldots,p\}$, the approximated IMSE is then
\begin{align}\label{hatIMSE}
\widehat\IMSE(J_n)=\tr\left[\left( \sum_{i\in J_n} \frac{1}{\ms_i^2}\, \phib_{m,i}\phib_{m,i}\TT + \Lambdab_m^{-1}\right)^{-1}\right] = \frac1N\, \tr\left[ \Cb \Mb^{-1}(\wb(J_n))\right] \,,
\end{align}
where $\Cb=\Lambdab_m$ and $\Mb(\wb)$ is defined by \eqref{M}, with $\Hb_i$ given by \eqref{Hi}, $\lambda=1/n$, $\ab_i= \Lambdab_m^{1/2}\phib_{m,i}/\ms_i$ for all $i$, and $w_i(J_n)$ equal to $1/n$ for $i\in J_n$ and zero otherwise.
We thus consider the minimization of $\tr\left[\Cb\Mb^{-1}(\wb)\right]$ for the construction of IMSE optimal designs.

We take $\mu$ uniform and $\mu_p(\ssb_i)=1/p$ for all $i$.
The covariance of the random field is given by Mat\'ern 3/2 kernel,
\begin{align*}
\mg(\ssb_1,\ssb_2)=(1+\mt\,\|\ssb_1-\ssb_2\|)\, \exp(-\mt\,\|\ssb_1-\ssb_2\|) \,,
\end{align*}
where we take $\mt=10$. The left panel of Figure~\ref{F:rho_k} shows the type of designs that are obtained. Here $\SSS_p$ corresponds to a $33\times 33$ regular grid in $\SSS=[0,1]^2$ and the truncation level $m$ equals 10. The construction relies on approximate design theory and we still need to extract an exact design (a finite set of support points) once optimal weights have been determined.
Several approaches can be used, see \citet{GP-CSDA2016, GPmoda11-2016}, which are beyond the scope of this paper. In particular, the method does not allow full control of the size of resulting design (which is of the order of magnitude of $m$). A straightforward approach is to simply remove points with negligible weights: on Figure~\ref{F:rho_k}-left, the 8 design points marked by a blue $\bullet$ have total weight less than $7.14\,10^{-5}$, the 12 other points are well spread over $\SSS$, with associated weights $w_i\in(0.0688, 0.0955)$. An important limitation of this approach comes from the need to diagonalize a $p\times p$ matrix, with $p$ being necessarily large when $d$ is large. A possible solution when $\mu$ is uniform on $\SX=[0,1]^d$ is to use a separable covariance so that eigenvalues (respectively, eigenfunctions) are products (respectively, tensor products) of their one-dimensional counterparts; see \citet{P-RESS2019}.
Alternatively, one can use a sparse kernel
induced by $k  \ll p$ points, which
reduces the size of the matrix to be diagonalized ($k\times k$ instead of $p \times p$); see~\citep{SagnolHW2016}
where this approach is used for the sequential design
of computer experiments.

In the rest of the example, $\SSS$ is the hypercube $[0,1]^d$ with $d=5$, $\SSS_p$ corresponds to the first $p$ points of Sobol' low-discrepancy sequence in $\SSS$, with $p=8\,192$. The truncation level is set at $m=50$. We consider the elimination of inessential $\Hb_i$, or equivalently of inessential $\ssb_i\in\SSS_p$, during the construction of an optimal design with the algorithm
\begin{align}\label{mult-Lopt}
w_i^{k+1} = \frac{w_i^k\, \ab_i\TT\Mb^{-1}(\wb^k)\Cb\Mb^{-1}(\wb^k)\ab_i}{\sum_{j=1}^p w_j^k\, \ab_j\TT\Mb^{-1}(\wb^k)\Cb\Mb^{-1}(\wb^k)\ab_j} \,,
\end{align}
initialized at $w_i^0=1/p$ for all $i$. As shown in \citep{PS2021-JSPI}, screening is computationally more efficient when applied periodically, every $\tau$ iterations, rather than at each $k$, and we use $\tau=100$ in the example. When inessential $\ssb_i$ are eliminated by a screening test performed with $\Mb(\wb^k)$, the weights of the remaining points, with indices in $I_k$, are renormalized into $w_i^k/\sum_{i\in I_k} w_i^k$ before application of \eqref{mult-Lopt}. The algorithm is stopped when $\delta<10^{-6}$, with $\delta$, given by \eqref{delta-L-opt}, measuring the (sub-)optimality of $\wb^k$.

We compare the bound
$D_2$
of Corollary~\ref{C:boundD2-L} (i.e., the bound $B$ of Corollary~\ref{C:bound-L-opt}) with the bounds
$B_1$ to $B_4$ given in \citep{PS2021-JSPI} for elimination of inessential points, both in terms of efficiency and computational time. Consider iteration $k$ where a screening test is performed, using the matrix $\Mb_k=\Mb(\wb^k)$.
$B_1$ and $B_2$
are consequences of the Equivalence Theorem (see Theorem~\ref{theo:ET} for the case of $c$-optimality) and require the calculation of the minimum and maximum eigenvalues of $\Mb_k^{-1/2}\Hb_i\Mb_k^{-1/2}$ for each $\Hb_i$ tested.
$B_3$ is derived from a second-order-cone-programming formulation of the design problem \citep{Sagnol2011}, its calculation requires the solution of one-dimensional convex minimization problem for each $\Hb_i$; we use the dichotomy line-search algorithm of \citet{PS2021-JSPI} with the precision parameter $\me$ fixed at 0.01.
$B_4$ corresponds to the method proposed in \citep{P_SPL-2013}\footnote{The method proposed there is for $A$-optimality, but since the matrix $\Cb$ in \eqref{hatIMSE} has full rank, the problem can be straightforwardly transformed into an $A$-optimal design problem.} and requires the calculation of the maximum eigenvalue of $\Mb_k$.

The right panel of Figure~\ref{F:rho_k} shows the proportion $\rho(k)$ of points eliminated along iterations for the five methods considered (note the staircase growth due to periodic screening every $\tau=100$ iterations only).
There is a rather clear global hierarchy in terms of elimination efficiency in this example, with
$B_3 \succ D_2 \succ B_2 \succ B_1 \succ B_4$ ($B_1$ eliminates more points than $B_2$ and $D_2$
in the early iterations, although this is not visible on the plot).
Since the cost of one iteration of \eqref{mult-Lopt} is roughly proportional to the size of $\wb$, that is, $p(1-\rho(k))$ at iteration $k$, we rescale the iteration counter $k$ into $C(k)=\sum_{i=1}^k (1-\rho(i))$ to count pseudo iterations that consider the decreasing size of $\wb^k$. The left panel of Figure~\ref{F:logdelta} presents $\delta$ (log scale) as a function of $C(k)$: it shows the acceleration provided by elimination of inessential points by each one of the method considered (compare with the curve with black $\Box$ for which $\rho(k) \equiv 0$) if one neglects the computational time of the screening tests. The hierarchy observed on Figure~\ref{F:rho_k}-right is confirmed. The right panel of Figure~\ref{F:logdelta} shows $\delta$ as a function of the true computational time\footnote{Calculations are in Matlab, on a PC with a clock speed of 2.5 GHz and 32 GB RAM --- only the comparison between curves is of interest.}. The most efficient method in terms of elimination,
$B_3$ (magenta $+$), yields the slowest convergence due to its high computational cost\footnote{The situation is reversed when $\ml=0$ in \eqref{Hi} and all $\Hb_i$ have rank 1
(so that neither $D_2$ nor $B_2$ can be used):
$B_3$ has an explicit expression and due to its high efficiency can provide an important acceleration factor; see \citet{PS2021-JSPI}.}.
Screening by
$D_2$ (red $\lozenge$), which does not require any eigenvalue calculation or numerical optimization, provides a speed-up factor of about 5 compared to the direct application of \eqref{mult-Lopt} without screening (black $\square$). Due to its reasonable computational cost,
$B_4$ (black $\times$) ensures a significant acceleration too despite its low screening efficiency. The poor performance of screening with $B_1$ (red $\bigstar$) and $B_2$
(blue $\triangledown$) is explained by the dimension $50 \times 50$ of the matrices $\Mb_k$. Notice that elimination of inessential points does not provide any speedup in terms of number of iterations required for a given accuracy: $18\,738$ iterations are needed to reach $\delta \leq 10^{-6}$ with or without screening (217 points have then a positive weight, among which 15 have a total mass less than $5\,10^{-5}$).

\begin{figure}[ht!]
\begin{center}
 \includegraphics[width=.53\linewidth]{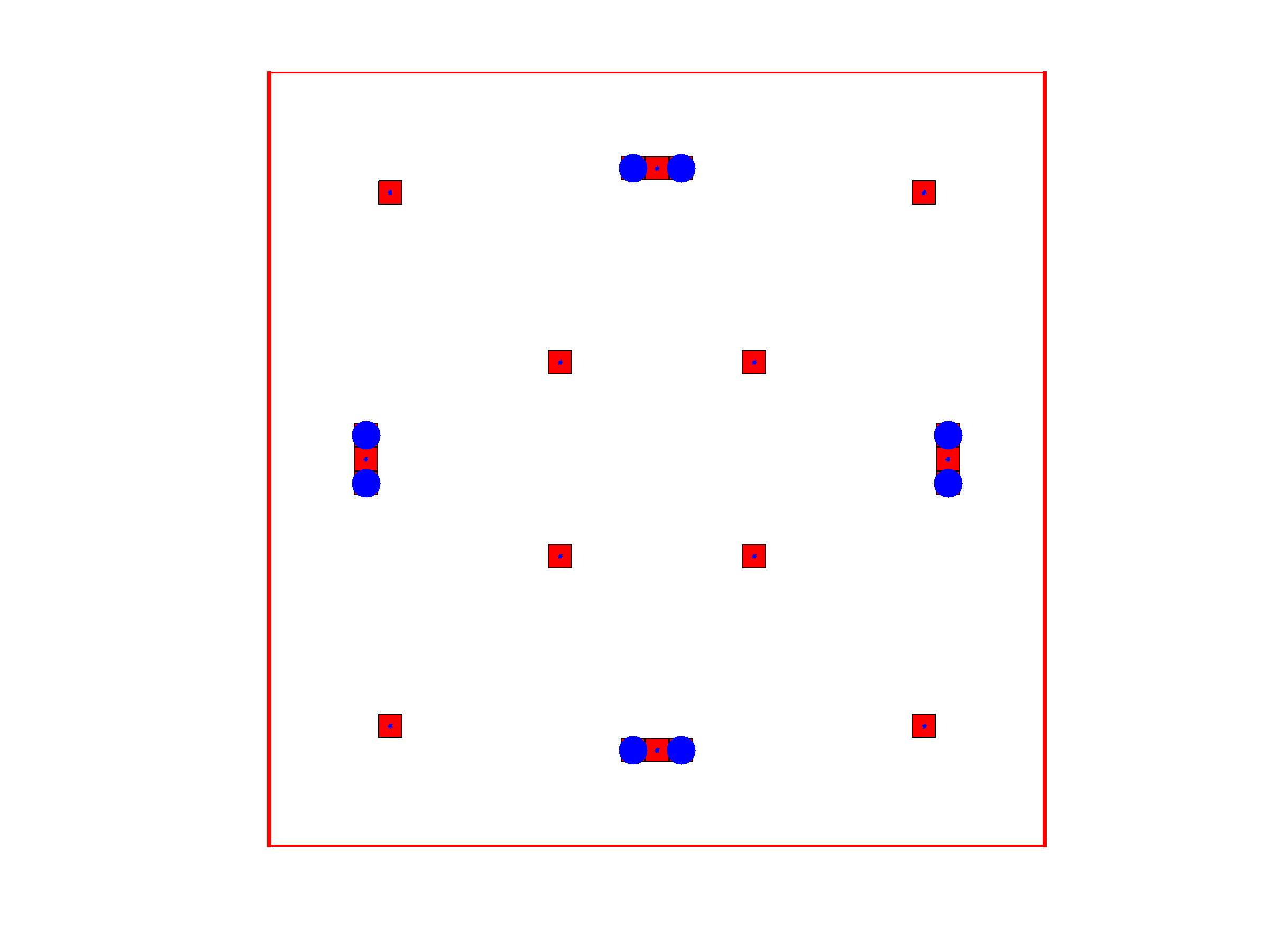}\hspace{-3em}
 \scalebox{0.4}{%% Creator: Inkscape inkscape 0.92.5, www.inkscape.org
%% Accompanies image file 'rho_k_q8192_m50_T100_B1234D2.pdf' (pdf, eps, ps)
\begingroup%
  \makeatletter%
  \providecommand\color[2][]{%
    \errmessage{(Inkscape) Color is used for the text in Inkscape, but the package 'color.sty' is not loaded}%
    \renewcommand\color[2][]{}%
  }%
  \providecommand\transparent[1]{%
    \errmessage{(Inkscape) Transparency is used (non-zero) for the text in Inkscape, but the package 'transparent.sty' is not loaded}%
    \renewcommand\transparent[1]{}%
  }%
  \providecommand\rotatebox[2]{#2}%
  \newcommand*\fsize{\dimexpr\f@size pt\relax}%
  \newcommand*\lineheight[1]{\fontsize{\fsize}{#1\fsize}\selectfont}%
  \ifx\svgwidth\undefined%
    \setlength{\unitlength}{576.48001099bp}%
    \ifx\svgscale\undefined%
      \relax%
    \else%
      \setlength{\unitlength}{\unitlength * \real{\svgscale}}%
    \fi%
  \else%
    \setlength{\unitlength}{\svgwidth}%
  \fi%
  \global\let\svgwidth\undefined%
  \global\let\svgscale\undefined%
  \makeatother%
  \begin{picture}(1,0.74979184)%
    \lineheight{1}%
    \setlength\tabcolsep{0pt}%
    \put(0,0){\includegraphics[width=\unitlength,page=1]{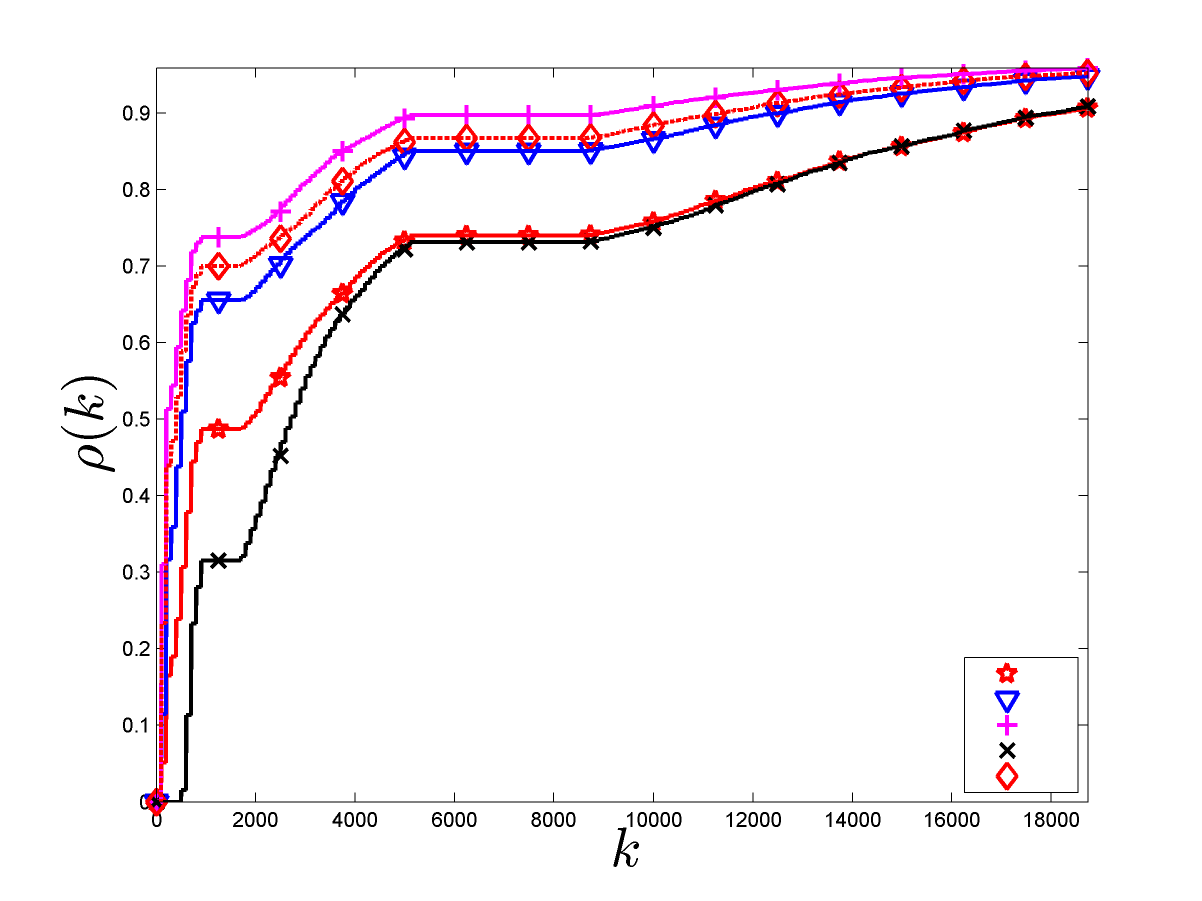}}%
    \put(0.86581779,0.17999646){\color[rgb]{0,0,0}\makebox(0,0)[lt]{\lineheight{1.25}\smash{\begin{tabular}[t]{l}$B_1$\end{tabular}}}}%
    \put(0.866592,0.15940448){\color[rgb]{0,0,0}\makebox(0,0)[lt]{\lineheight{1.25}\smash{\begin{tabular}[t]{l}$B_2$\end{tabular}}}}%
    \put(0.86571424,0.13964893){\color[rgb]{0,0,0}\makebox(0,0)[lt]{\lineheight{1.25}\smash{\begin{tabular}[t]{l}$B_3$\end{tabular}}}}%
    \put(0.86572202,0.11734708){\color[rgb]{0,0,0}\makebox(0,0)[lt]{\lineheight{1.25}\smash{\begin{tabular}[t]{l}$B_4$\end{tabular}}}}%
    \put(0.8657221,0.09629502){\color[rgb]{0,0,0}\makebox(0,0)[lt]{\lineheight{1.25}\smash{\begin{tabular}[t]{l}$D_2$\end{tabular}}}}%
  \end{picture}%
\endgroup%
}
\end{center}
\caption{\small Left: Design obtained when $\SSS_p$ is a $33\times 33$ regular grid in $\SSS=[0,1]^2$, $m=10$ (the 8 points marked with a blue $\bullet$ have negligible weight). Right: evolution of the proportion $\rho(k)$ of points eliminated by the five methods considered along iterations ($d=5$, $p=8\,192$, $m=50$).}
\label{F:rho_k}
\end{figure}
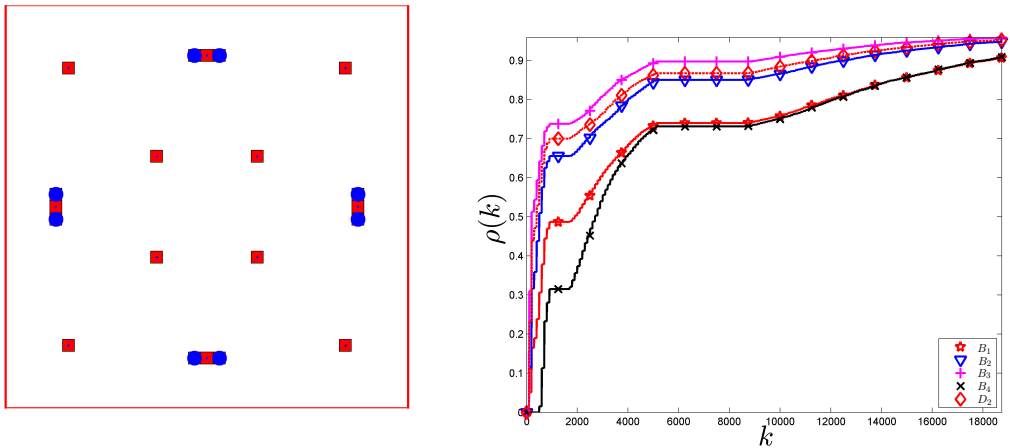

\begin{figure}[ht!]
\begin{center}
 \scalebox{0.39}{%% Creator: Inkscape inkscape 0.92.5, www.inkscape.org
%% Accompanies image file 'logdelta_pseudok_q8192_m50_T100_B1234D2.pdf' (pdf, eps, ps)
\begingroup%
  \makeatletter%
  \providecommand\color[2][]{%
    \errmessage{(Inkscape) Color is used for the text in Inkscape, but the package 'color.sty' is not loaded}%
    \renewcommand\color[2][]{}%
  }%
  \providecommand\transparent[1]{%
    \errmessage{(Inkscape) Transparency is used (non-zero) for the text in Inkscape, but the package 'transparent.sty' is not loaded}%
    \renewcommand\transparent[1]{}%
  }%
  \providecommand\rotatebox[2]{#2}%
  \newcommand*\fsize{\dimexpr\f@size pt\relax}%
  \newcommand*\lineheight[1]{\fontsize{\fsize}{#1\fsize}\selectfont}%
  \ifx\svgwidth\undefined%
    \setlength{\unitlength}{576.48001099bp}%
    \ifx\svgscale\undefined%
      \relax%
    \else%
      \setlength{\unitlength}{\unitlength * \real{\svgscale}}%
    \fi%
  \else%
    \setlength{\unitlength}{\svgwidth}%
  \fi%
  \global\let\svgwidth\undefined%
  \global\let\svgscale\undefined%
  \makeatother%
  \begin{picture}(1,0.74979184)%
    \lineheight{1}%
    \setlength\tabcolsep{0pt}%
    \put(0,0){\includegraphics[width=\unitlength,page=1]{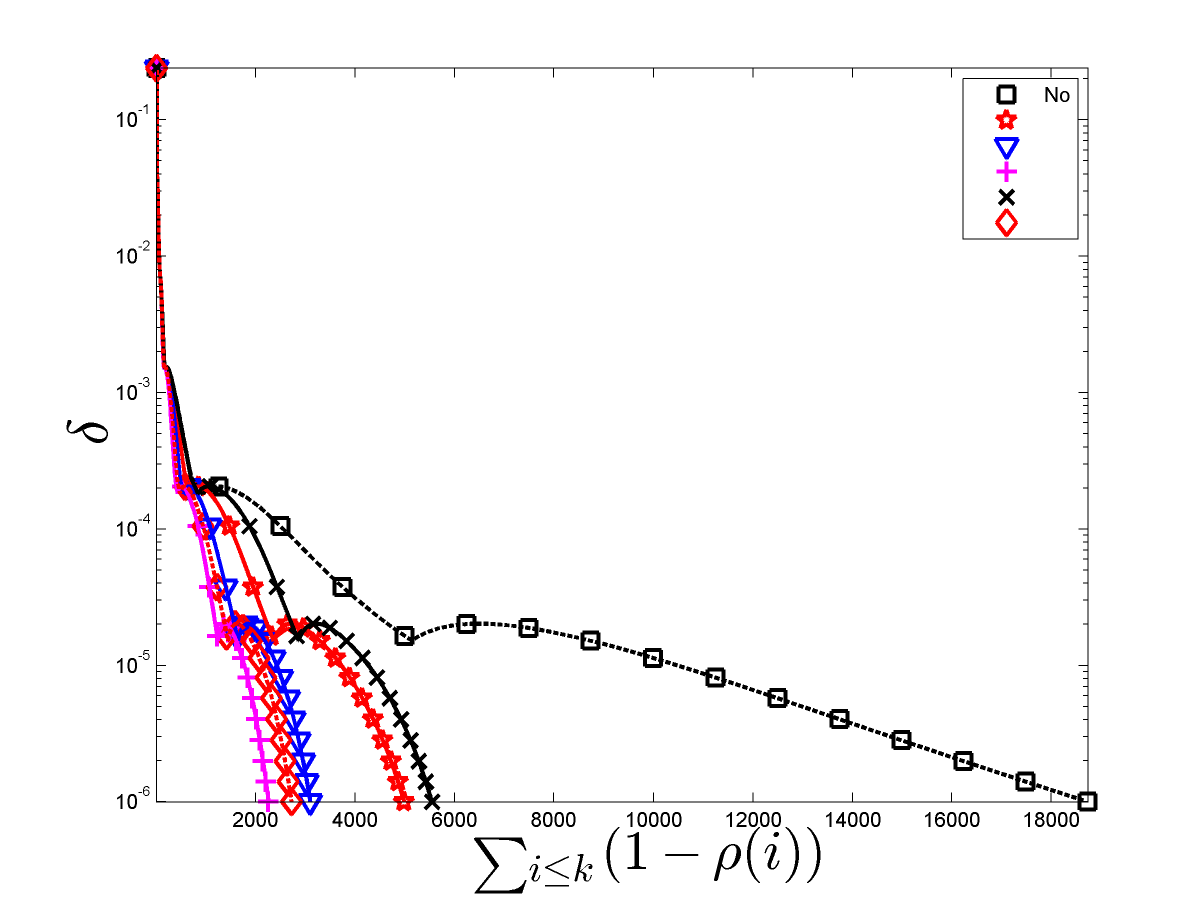}}%
    \put(0.86772522,0.64146071){\color[rgb]{0,0,0}\makebox(0,0)[lt]{\lineheight{1.25}\smash{\begin{tabular}[t]{l}$B_1$\end{tabular}}}}%
    \put(0.86849943,0.62086873){\color[rgb]{0,0,0}\makebox(0,0)[lt]{\lineheight{1.25}\smash{\begin{tabular}[t]{l}$B_2$\end{tabular}}}}%
    \put(0.86762167,0.60111319){\color[rgb]{0,0,0}\makebox(0,0)[lt]{\lineheight{1.25}\smash{\begin{tabular}[t]{l}$B_3$\end{tabular}}}}%
    \put(0.86762945,0.57881133){\color[rgb]{0,0,0}\makebox(0,0)[lt]{\lineheight{1.25}\smash{\begin{tabular}[t]{l}$B_4$\end{tabular}}}}%
    \put(0.86762953,0.55775927){\color[rgb]{0,0,0}\makebox(0,0)[lt]{\lineheight{1.25}\smash{\begin{tabular}[t]{l}$D_2$\end{tabular}}}}%
  \end{picture}%
\endgroup%
}\hspace{-2em}
 \scalebox{0.39}{%% Creator: Inkscape inkscape 0.92.5, www.inkscape.org
%% Accompanies image file 'logdelta_T_q8192_m50_T100_B1234D2.pdf' (pdf, eps, ps)
\begingroup%
  \makeatletter%
  \providecommand\color[2][]{%
    \errmessage{(Inkscape) Color is used for the text in Inkscape, but the package 'color.sty' is not loaded}%
    \renewcommand\color[2][]{}%
  }%
  \providecommand\transparent[1]{%
    \errmessage{(Inkscape) Transparency is used (non-zero) for the text in Inkscape, but the package 'transparent.sty' is not loaded}%
    \renewcommand\transparent[1]{}%
  }%
  \providecommand\rotatebox[2]{#2}%
  \newcommand*\fsize{\dimexpr\f@size pt\relax}%
  \newcommand*\lineheight[1]{\fontsize{\fsize}{#1\fsize}\selectfont}%
  \ifx\svgwidth\undefined%
    \setlength{\unitlength}{576.48001099bp}%
    \ifx\svgscale\undefined%
      \relax%
    \else%
      \setlength{\unitlength}{\unitlength * \real{\svgscale}}%
    \fi%
  \else%
    \setlength{\unitlength}{\svgwidth}%
  \fi%
  \global\let\svgwidth\undefined%
  \global\let\svgscale\undefined%
  \makeatother%
  \begin{picture}(1,0.74979184)%
    \lineheight{1}%
    \setlength\tabcolsep{0pt}%
    \put(0,0){\includegraphics[width=\unitlength,page=1]{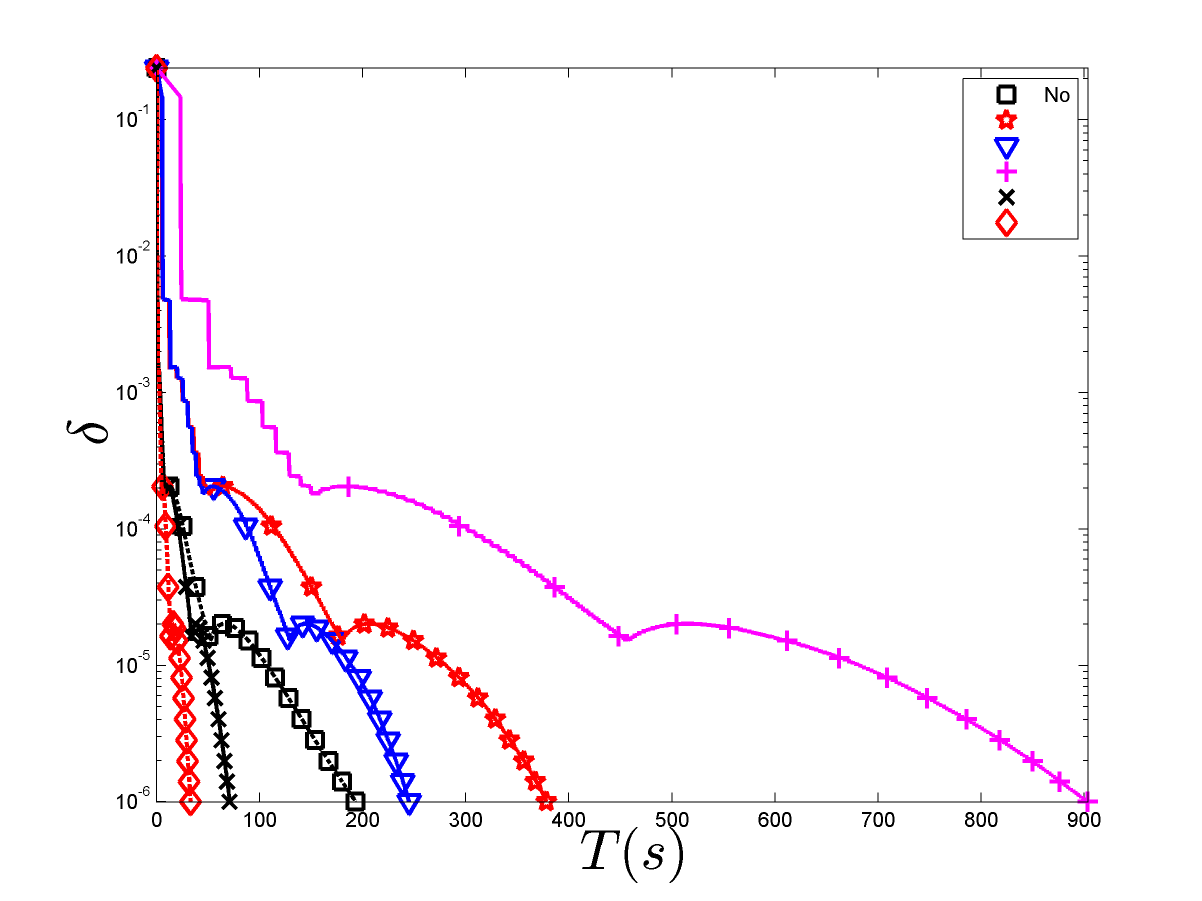}}%
    \put(0.86715571,0.64290431){\color[rgb]{0,0,0}\makebox(0,0)[lt]{\lineheight{1.25}\smash{\begin{tabular}[t]{l}$B_1$\end{tabular}}}}%
    \put(0.86792993,0.62231233){\color[rgb]{0,0,0}\makebox(0,0)[lt]{\lineheight{1.25}\smash{\begin{tabular}[t]{l}$B_2$\end{tabular}}}}%
    \put(0.86705216,0.60255682){\color[rgb]{0,0,0}\makebox(0,0)[lt]{\lineheight{1.25}\smash{\begin{tabular}[t]{l}$B_3$\end{tabular}}}}%
    \put(0.86705995,0.58025495){\color[rgb]{0,0,0}\makebox(0,0)[lt]{\lineheight{1.25}\smash{\begin{tabular}[t]{l}$B_4$\end{tabular}}}}%
    \put(0.86706003,0.55920289){\color[rgb]{0,0,0}\makebox(0,0)[lt]{\lineheight{1.25}\smash{\begin{tabular}[t]{l}$D_2$\end{tabular}}}}%
  \end{picture}%
\endgroup%
}
\end{center}
\caption{\small Left: Evolution of $\delta$ given by \eqref{delta-L-opt} as a function of the pseudo iteration number (the computational cost of each screening test is ignored). Right: Evolution of $\delta$ as a function of computational time (in seconds).}
\label{F:logdelta}
\end{figure}
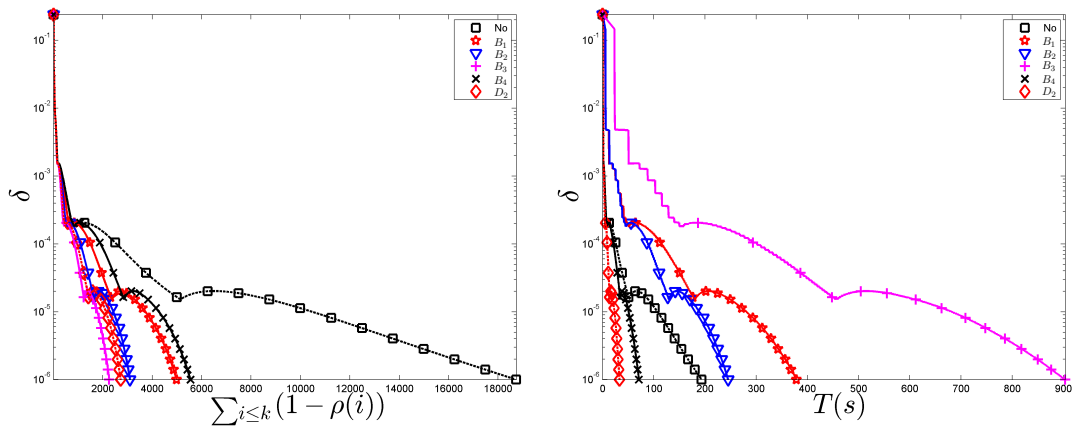

The hierarchy between methods in terms of computational time depends on $m$ and may vary with the choice of the kernel $\mg(\cdot,\cdot)$ (in particular, with its correlation length), but in this example the ranking was fairly stable when varying $\mg$ (we also used the Mat\'ern 1/2 and 5/2 kernels and the Gaussian, i.e., squared exponential, kernel) and the number $p$ of Sobol' points in $\SSS_p$.

%_____________________________________________________________________________________________________________
\section{Conclusion}\label{S:conclusion}

We have shown that the reformulation of $c$- and $L$-optimal design problems as a linear regression problem with squared sparsity-inducing penalty can be used to obtain safe
screening rules. These rules are much faster to compute than those previously proposed in the optimal design community, especially for problems with a high dimensional parameter space. Their efficiency has been demonstrated on several experiments using real data sets. In addition, we have leveraged the polyhedral conic geometry
of the dual quadratic-lasso problem to develop a new homotopy algorithm for Bayesian $c$-optimal design. This algorithm terminates after a finite number of steps with
an optimal design and can be several orders of magnitude faster than standard first-order algorithms.
%____________________________________________________________________________________________

{\small

%___________________________________________________________________________________________
\section*{Acknowledgements}
The authors thank three anonymous referees whose comments contributed to improve the presentation of these results.

The work of Luc Pronzato was partly supported by project INDEX (INcremental Design of EXperiments) ANR-18-CE91-0007 of the French National Research Agency (ANR).
The work of Guillaume Sagnol is supported by
the Deutsche Forschungsgemeinschaft (DFG, German Research Foundation) under Germany's Excellence Strategy --- The Berlin Mathematics Research Center MATH+ (EXC-2046/1, project ID: 390685689).
%____________________________________________________________________________________________

\vskip 0.2in
%\bibliography{Lasso_jmlr}

% <-----   content of bbl file

%  end of bbl file  -------->

%--------------------------------------------------------------------------------------------

\newpage
\appendix

%\section{\texorpdfstring{$L$}{L}-optimality and quadratic group Lasso}\label{S:Lopt}
\section{$L$-optimality and quadratic group Lasso}\label{S:Lopt}

In order to generalize the results of Section~\ref{sec:Qlasso} to the case of $L$-optimality, we start with the counterpart of Theorem~\ref{theo:eq-copt-qlasso} for matrices $\Hb_i$ of the form $\Hb_i=\Ab_i\Ab_i\TT + \lambda\, \Ib$ for some $\Ab_i\in\mathds{R}^{m \times q}$,
for all $i\in[p]$. This result is interesting in its own right, as this situation arises naturally for the linear model with multiresponse experiments, where~\eqref{regression1} is replaced by
a $q$-dimensional observation $Y_i=\Ab_i\TT \mtb + \mveb_i$.

\begin{theo}\label{theo:eq-multi-copt-qlasso}
Let $\Hb_i=\Ab_i\Ab_i\TT + \lambda\, \Ib$, with $\Ab_i\in\mathds{R}^{m \times q}$. Then, the $c$-optimal design problem~\eqref{phi} is \emph{equivalent} to the following problem:
\begin{align}\label{mqlasso}
\min_{\xb_1,\ldots,\xb_p \in \mathds{R}^{q}}\quad \SL^q _\lambda(\xb)= \left\| \sum_{i=1}^p \Ab_i \xb_i-\cb \right\|^2 + \lambda\, \left(\sum_{i=1}^p \|\xb_i\| \right)^2,
\end{align}
in the following sense:
\begin{itemize}
 \item[(i)] the optimal value of~\eqref{mqlasso} is equal to $\lambda\, \phi_c(\wb^*)$, where $\wb^*$ is a $c$-optimal design;
 \item[(ii)] If $\xb^*=(\xb_1^*,\ldots,\xb_p^*)$ solves~\eqref{mqlasso}
 and $\Ab_i\TT\cb\neq \0b$ for some $i\in[p]$,
then $\xb^*\neq \0b$ and
 $\widehat\wb^*=\widehat\wb(\xb^*)$ is $c$-optimal, where
\begin{align}\label{mhatw}
\widehat w_i(\xb)=\frac{\|\xb_i\|}{\sum_{j=1}^p \|\xb_j\|}\,, \ \forall i \in[p]\,, \ \xb\neq\0b\,;
\end{align}
 \item[(iii)] If $\wb^*\in\SP_{p}$ is c-optimal, then
 $\widehat\xb^*=(\widehat\xb_1(\wb^*),\ldots,\widehat\xb_p(\wb^*))$ is optimal for~\eqref{mqlasso}, where
\begin{align}
\widehat \xb_i(\wb)= w_i\, \Ab_i\TT \Mb^{-1}(\wb)\cb\,, \ \forall i\in[p]\,. \label{mwidehat-x}
 \end{align}
 \item[(iv)] In the pathological case $\Ab_i\TT\cb=\0b$, $\forall i\in[p]$, the unique optimal solution to the quadratic lasso is $\xb^*=\0b$, while every design $\wb\in\SP_p$ is $c$-optimal.
 \end{itemize}
\end{theo}

\begin{proof}
The proof follows almost the same lines as the proof of Theorem~\ref{theo:eq-copt-qlasso}, with the difference that we define
$\Ab = [\Ab_1,\ldots,\Ab_p] \in \mathds{R}^{m \times (pq)}$
and
$\Db(\wb) = \diag(\wb) \otimes \Ib_q \in \mathds{R}^{pq \times pq }$, i.e., the diagonal matrix with
$p$ diagonal blocks of the form $w_i \Ib_q$.
Then, we still have $\Mb(\wb) = \Ab \Db(\wb) \Ab\TT + \lambda \Ib_{m}$.
We also change the definition of the function $v$ to
\[v(\xb,\wb)=\left\|\sum_{i=1}^p\Ab_i\xb_i-\cb\right\|^2 + \ml \sum_{i=1}^p
\frac{\|\xb_i\|^2}{w_i}
=\|\Ab\xb-\cb\|^2 + \ml\, \xb\TT \Db^{-1}(\wb)\xb
.\]
Here, the function $(\xb_i,w_i) \mapsto\frac{\|\xb_i\|^2}{w_i}$ is the perspective function of $\xb_i\mapsto\|\xb_i\|^2$, hence it is convex, and defined by continuity in $w_i=0$ (i.e., $\|\xb_i\|^2/0=0$ if $\xb_i=\0b$ and $\|\xb_i\|^2/0=\infty$ otherwise).

The rest of the proof is identical to that of Theorem~\ref{theo:eq-copt-qlasso}. We obtain
that~\eqref{minv_w} and~\eqref{minv_x} hold,
which imply $(i)$, $(ii)$ and $(iii)$.
That $\xb=\0b$ cannot be optimal whenever $\Ab\TT\cb\neq\0b$ follows from
 \[
 \SL_\lambda(\frac{\ab_k\TT\cb}{\|\ab_k\|^2+\lambda} \eb_k)=
 \|\cb\|^2 - \frac{(\ab_k\TT\cb)^2}{\|\ab_k\|^2+\lambda}
< \|\cb\|^2 = \SL_{\ml}(\0b),
\]
if $k$ is such that the
$k$th column $\ab_k$ of $\Ab$ satisfy $\ab_k\TT\cb\neq 0$. The proof of point $(iv)$ for the case $\Ab\TT\cb=\0b$ is also the same
as for Theorem~\ref{theo:eq-copt-qlasso}.
\end{proof}

Let us return to the case where $\Hb_i=\ab_i\ab_i\TT + \lambda \Ib_m$.
The generalization of Theorem~\ref{theo:eq-copt-qlasso} to the case of $L$-optimality is now straightforward following
%if we use the fact that the $L-$optimality criterion can be
%rewritten as a $c-$optimality criterion for an
%extended model with replicated observations of multiple copies of the unknown parameter~
\citep{Sagnol2011}. Let $\Kb\in\mathds{R}^{m\times r}$ and define $\cbb=\operatorname{vec}(\Kb)\in\mathds{R}^{mr}$. Then,
for all $\wb \in \SP_p$
we have
\[
\phi_L(\wb)  = \Phi_L[\Mb(\wb)] = \tr[\Kb\TT \Mb^{-1}(\wb)\, \Kb]\,
=
\cbb\TT \widetilde \Mb^{-1}(\wb) \cbb,
\]
where
\[
\widetilde \Mb(\wb) = \Ib_r \otimes \Mb(\wb) =
\sum_{i=1}^p w_i (\Ib_r \otimes (\ab_i \ab_i\TT)) + \lambda \Ib_{rm}
= \sum_{i=1}^p w_i\ \widetilde \Hb_i \,,
\]
and $\widetilde \Hb_i=(\Ib_r \otimes \ab_i)(\Ib_r \otimes \ab_i)\TT + \lambda \Ib_{rm}$.
Note that $\widetilde\Hb_i$ has the form $\widetilde \Ab_i \widetilde\Ab_i\TT+\lambda\Ib$ with $\widetilde\Ab_i=\Ib_r\otimes \ab_i\in\mathds{R}^{rm \times r}$.
Theorem~\ref{theo:eq-multi-copt-qlasso} can thus be applied to $L$-optimality. Straightforward manipulations allows one to get rid off Kronecker products by reshaping the vector $\xb\in\mathds{R}^{pr}$ into a matrix $\Xb\in\mathds{R}^{p \times r}$, which yields the following result.

\begin{theo}\label{theo:eq-Lopt-qlasso}
Let $\Hb_i$ be given by~\eqref{Hi} and denote by $\Ab$ the $m \times p$-matrix with columns $(\ab_i)_{i\in [p]}$. Consider the $L$-optimal design problem
where we minimize $\phi_L(\wb)$ given by~\eqref{phiL} with respect to $\wb\in\SP_{p}$ with $\Kb$ an $m \times r$ matrix. This problem is \emph{equivalent} to the following one, which we call \emph{quadratic group Lasso},
\begin{align}\label{qglasso}
\min_{\Xb \in \mathds{R}^{p \times r}}\quad
\SL_{\lambda}'(\Xb)=
\| \Ab\Xb-\Kb\|_F^2 + \lambda\, (\|\Xb\|_{1,2})^2 \,,
\end{align}
in the following sense:
\begin{itemize}
 \item[(i)] the optimal value of~\eqref{qglasso} is equal to $\lambda\, \phi_L(\wb^*)$, where $\wb^*$ is an $L$-optimal design;
 \item[(ii)] If $\Xb^*$ solves the quadratic group Lasso problem and $\Ab\TT\Kb\neq\Ob$, then $\widehat\wb^*=\widehat\wb(\Xb^*)$ is $L$-optimal, where
\begin{align}\label{hatwL}
\widehat w_i(\Xb)=\frac{\|\Xb_{i,\cdot}\|}{\|\Xb\|_{1,2}}\,\ \forall i \in [p]\,, \ \Xb\neq\Ob\,;
\end{align}
 \item[(iii)] If $\wb^*\in\SP_{p}$ is $L$-optimal, then $\widehat\Xb^*=\widehat\Xb(\wb^*)$ is optimal for~\eqref{qglasso}, where
\begin{align}
\widehat \Xb_{i,\cdot}(\wb)= w_i\, \ab_i\TT \Mb^{-1}(\wb)\Kb\,. \label{widehat-xL}
 \end{align}
 \item[(iv)] In the pathological case $\Ab\TT\Kb=\Ob$, the unique optimal solution to the quadratic lasso is $\xb^*=\Ob$, while every design $\wb\in\SP_p$ is $c$-optimal.
 \end{itemize}
\end{theo}

\medskip
Proceeding as above, we can also show that the dual of the quadratic group Lasso problem can be written as
\begin{align}\label{dualL}
\max_{\yb\in\mathds{R}^m}\ \SD_\ml'(\Yb) = \|\Kb\|_F^2-\|\Yb-\Kb\|_F^2 - \frac{\max_{i\in [p]}\|\Yb\TT\ab_i\|^2}{\ml},
\end{align}
and that the following relations between optimal primal and dual variables hold:
\begin{align*}
\Yb^*=\Kb-\Ab\Xb^*=\lambda \Mb^{-1}(\wb^*) \Kb \,,
\end{align*}
with $\wb^*$ an $L$-optimal design.
The generalization of Theorem~\ref{theo:B0} is not completely straightforward, though. We first prove the following lemma:

\begin{lemm}\label{L:CSExt}
Let $\Yb\in \mathds{R}^{m\times r}$ and $\ab\in \mathds{R}^m$. Then,
\begin{align*}
\sup_{\Zb\in \mathds{R}^{m\times r}} \{\|\Zb^{\top}\ab\|:\ \|\Yb-\Zb\|_F\leq R\} =   \|\Yb\TT\ab\| + R\, \|\ab\| \,,
\end{align*}
the supremum being reached for $\Zb = \Zb_*= \Yb + R\,\ab\ab\TT\Yb/(\|\ab\| \, \|\Yb\TT\ab\|)$.
\end{lemm}

\begin{proof}
 Denote by $\yb_i$ and $\zb_i$ the $i$-th column of $\Yb$ and $\Zb$, respectively, and define $u_i=\|\yb_i-\zb_i\|$, $i\in[r]$, and $\ub=(u_1,\ldots,u_r)\TT$. Then, $|\zb_i\TT\ab| \leq |\yb_i\TT\ab|+u_i\,\|\ab\|$, $i\in[r]$. Therefore,
\begin{align*}
\|\Zb\TT \ab\|^2 = \sum_{i=1}^r (\zb_i\TT \ab)^2 &\leq \sum_i (|\yb_i\TT \ab| + u_i \|\ab\|)^2 = \|\Yb\TT\ab\|^2 + \|\ub\|^2 \, \|\ab\|^2 + 2 \|\ab\| \sum_{i=1}^r
u_i\, |\yb_i\TT \ab|\,\\
&\leq \|\Yb\TT\ab\|^2 + \|\ub\|^2 \, \|\ab\|^2 + 2 \|\ab\|\,\|\ub\|\,\|\Yb\TT\ab\| \\
&\leq \|\Yb\TT\ab\|^2 + R^2\, |\ab\|^2 + 2\,R\,\|\ab\|\,\|\Yb\TT\ab\| = (\|\Yb\TT\ab\| + R \|\ab\|)^2 \,,
\end{align*}
where we used the Cauchy-Schwarz inequality and $\|\ub\|^2=\|\Yb-\Zb\|_F\leq R$. The supremum is reached when $u_i=\ma |\yb_i\TT \ab|$, $i\in[r]$ for some $\ma>0$ and $\zb_i-\yb_i=\beta_i\, \ab$, i.e., $\Zb=\Yb+\ab\betab\TT$ for some $\betab=(\beta_1,\ldots,\beta_r)\TT\in\mathds{R}^r$. Direct calculation shows that it implies
$\Zb = \Zb_*$.
\end{proof}

\begin{theo}\label{theo:B0L}
Let $\Yb\in\mathds{R}^{m\times r}$ be an $\me$-suboptimal dual solution,
i.e., $\SD_\lambda'(\Yb)\geq\SD_\lambda'(\Yb^*)-\epsilon$.
Then, any $\ab_i$ satisfying
\begin{align*}
D_0(\ab_i; \Yb, \epsilon)=\max_{j\in [p]} \|\Yb\TT\ab_j\| - \|\Yb\TT\ab_i\|  -\sqrt{\epsilon (\|\ab_i\|^2+\lambda)}>0
\end{align*}
cannot support an optimal design.
\end{theo}

\begin{proof}
Let $u=\max_{i\in [p]}\|\Yb\TT\ab_i\|/\sqrt{\ml}$ and $u^*=\max_{i\in [p]}\|\Yb^{*\top}\ab_i\|/\sqrt{\ml}$. Following the same lines as in the proof of Theorem~\ref{theo:B0}, we get $\|\Yb-\Yb^*\|_F^2 + (u-u^*)^2\leq \epsilon$.

Now, the complementary slackness condition of optimality of \eqref{qglasso} and \eqref{dualL} indicates that the $i$-th row of $\Xb^*$ vanishes whenever
$\|\Yb^{*\top}\ab_i\| < \sqrt{\lambda}u^*$, and thus $\ab_i$ is inessential. Using $\|\Yb-\Yb^*\|_F\leq  \sqrt{\epsilon-(u-u^*)^2}$ and Lemma~\ref{L:CSExt}, we get
\begin{align*}
\|\Yb^{*\top}\ab_i\| - \sqrt{\lambda}u^*
&\leq \sup_{u'}\quad  -\sqrt{\lambda} u' + \sup_{\Yb'\in \mathds{R}^{m\times r}} \{\|\Yb^{'\top}\ab_i\|:\ \|\Yb-\Yb'\|_F\leq \sqrt{\epsilon-(u-u')^2}\}\\
&= \sup_{u'}\quad -\sqrt{\lambda} u'+ \|\Yb^{\top}\ab_i\|
+ \sqrt{\epsilon-(u-u')^2}\, \|\ab_i\| \,.
\end{align*}
Simple calculations show that the supremum is attained at $u'=u\pm \sqrt{\lambda \epsilon}/\sqrt{\|\ab_i\|^2+\lambda}$, and substitution yields
\begin{align*}
\|\Yb^{*\top}\ab_i\| - \sqrt{\lambda}u^* \leq
\|\Yb\TT\ab_i\| - \sqrt{\lambda} u +\sqrt{\epsilon (\|\ab_i\|^2+\lambda)}\,.
\end{align*}
Therefore, $\ab_i$ is inessential if $\|\Yb\TT\ab_i\| -\max_{j\in [p]}\|\Yb\TT\ab_j\|  +\sqrt{\epsilon (\|\ab_i\|^2+\lambda)}<0$.
\end{proof}

As in Section~\ref{S:rules-r=1}, we obtain two screening rules for $L$-optimality by considering the dual points $\Yb_1(\Xb)=\Kb-\Ab\Xb$ and $\Yb_2(\widehat\wb(\Xb))=\lambda \Mb^{-1}(\widehat\wb(\Xb)) \Kb$, where
\begin{align*}
\Mb(\widehat\wb(\Xb)) = \frac{1}{\|\Xb\|_{1,2}}\ \sum_{i=1}^p \|\Xb_{i,\cdot}\|\ \ab_i \ab_i\TT + \lambda \Ib_m \,.
\end{align*}

\begin{coro} \label{C:boundD2-L}
Let $\Xb\in\mathds{R}^{p \times r}$. If
 \begin{align*}
 D_1(\ab_i;\Xb)= D_0\Big(\ab_i;\, \Yb_1(\Xb),\ \SL'_\ml(\Xb)-\SD'_\ml[\Yb_1(\Xb)] \Big) > 0,
 \end{align*}
or $\Xb \neq \Ob$ and
 \begin{align*}
 D_2(\ab_i;\Xb)= D_0\Big(\ab_i;\, \Yb_2(\widehat\wb(\Xb)),\ \ml\, \tr[\Kb\TT\Mb^{-1}(\widehat\wb(\Xb))\Kb] -\SD'_\ml[\Yb_2(\widehat\wb(\Xb))] \Big) > 0,
 \end{align*}
then $\ab_i$ cannot support an $L$-optimal design.
\end{coro}

As in Section~\ref{S:rules-r=1}, in a design context one may start with a given $\wb$ and use $\Xb=\widehat\Xb(\wb)$ in $D_1$ and $\widehat\wb(\Xb)=\wb$ in $D_2$. Similarly to Lemma~\ref{L:y1=y2}, we have
$\Yb_1(\widehat\Xb(\wb))=\Yb_2(\wb)$
and we
obtain the following generalization of Corollary~\ref{C:bound-c-opt}.

\begin{coro}\label{C:bound-L-opt}
Let $\Mb$ be any matrix in the convex hull of $\mathds{H}$ with the $\Hb_i$ satisfying \eqref{Hi}. Then, any matrix $\Hb_i$ in $\mathds{H}$ such that $B(\Mb,\Hb_i)>0$, where
\begin{align}
B(\Mb,\Hb_i) &= \left\{(1+\delta)\, \Phi_L(\Mb)-\ml\,\tr(\Kb\TT\Mb^{-2}\Kb)\right\}^{1/2} \nonumber \\
&\qquad - \left[\delta\, \Phi_L(\Mb)\left(1+ \frac{\|\ab_i\|^2}{\ml}\right)\right]^{1/2}  - |\Kb\TT\Mb^{-1}\ab_i| \label{bound-L-opt}
\end{align}
and
\begin{align}\label{delta-L-opt}
\delta=\max_{\Hb_i\in\mathds{H}} \frac{\tr(\Kb\TT\Mb^{-1}\Hb_i\Mb^{-1}\Kb)}{\tr(\Kb\TT\Mb^{-1}\Kb)}-1 \,,
\end{align}
does not support a $c$-optimal design.
\end{coro}

We have $\delta\geq0$ and $\Phi_L(\Mb)\leq (1+\delta) \phi_L(\wb^*)$ with $\wb^*$ an $L$-optimal design. 
}

\end{document}